\documentclass[lettersize,journal]{IEEEtran}
\usepackage{amsmath,amsfonts}
\usepackage{algorithmic}
\usepackage{algorithm}
\usepackage{array}
\usepackage[caption=false,font=footnotesize,labelfont=rm,textfont=rm]{subfig}
\usepackage{textcomp}
\usepackage{stfloats}
\usepackage{url}
\usepackage{verbatim}
\usepackage{graphicx}
\usepackage{cite}
\usepackage{color}
\usepackage[colorlinks=true, allcolors=blue]{hyperref}

\let\oldeqref\eqref 
\renewcommand{\eqref}[1]{\textcolor{blue}{\oldeqref{#1}}} 
\usepackage{amssymb}
\usepackage{makecell}
\usepackage{amsthm}
\usepackage{bm}
\usepackage{balance}

\newtheoremstyle{mystyle}
{0pt}
{0pt}
{\normalfont}
{\parindent}
{\normalfont\itshape}
{:}
{0.5em}
{\thmname{#1}\thmnumber{\itshape\space #2}\thmnote{\space (#3)}}
\theoremstyle{mystyle}

\newtheorem{corollary}{Corollary}

\newtheorem{lemma}{Lemma}
\newtheorem{remark}{Remark}

\makeatletter
\renewenvironment{proof}[1][\proofname]{\par
	\pushQED{\qed}%
	\normalfont
	\topsep0\p@\relax
	\trivlist
	\item[\hskip\labelsep
	\hspace{0.65em}
	\normalfont\itshape
	#1\@addpunct{:}]
	\normalfont\ignorespaces
}{
	\popQED
	\endtrivlist
	\@endpefalse
}
\makeatother

\makeatother
\usepackage{multirow}
\usepackage{diagbox}
\usepackage{threeparttable}
\newcommand{\algorithmicparameter}{\textbf{ Setup:}}
\newcommand{\PARAMETER}{\item[\algorithmicparameter]}

\renewcommand{\algorithmicparameter}{\textbf{Setup:}}

\usepackage[caption=false,font=footnotesize,labelfont=rm,textfont=rm]{subfig}
\begin{document}
	\title{LoRa Fluid Antenna Multiple Access}

	\author{
	Gaoze Mu,~\IEEEmembership{Member,~IEEE,} Yanzhao Hou,~\IEEEmembership{Member,~IEEE,} Peichang Zhang,~\IEEEmembership{Member,~IEEE,} Mingjie~Chen, Siyuan~Li, Qimei~Cui, \IEEEmembership{Senior Member,~IEEE,} Xiaofeng Tao, \IEEEmembership{Senior Member,~IEEE}
	
	\thanks{(\textit{Corresponding author: Yanzhao Hou).}}
	
	\thanks{Gaoze Mu, Mingjie Chen, and Siyuan Li are with the National Engineering Research Center for Mobile Network Technologies, Beijing University of Posts and Telecommunications, Beijing 100876, China (e-mail: \{mugz; chenmingjie; lisiyuan\}@bupt.edu.cn).}
	
	\thanks{Yanzhao Hou, Qimei Cui, and Xiaofeng Tao are with the National Engineering Research Center for Mobile Network Technologies, Beijing University of Posts and Telecommunications, Beijing 100876, China, and also with the Department of Broadband Communication, Peng Cheng Laboratory, Shenzhen 518055, China (e-mail: \{houyanzhao; cuiqimei; taoxf\}@bupt.edu.cn).}

	\thanks{Peichang Zhang are with the State Key
	Laboratory of Radio Frequency Heterogeneous Integration, Shenzhen University, Shenzhen 518060, China (e-mail: pzhang@szu.edu.cn).}

}

\markboth{}
{Shell \MakeLowercase{\textit{et al.}}: A Sample Article Using IEEEtran.cls for IEEE Journals}


\maketitle
\begin{abstract}
	Concurrent long-range (LoRa) transmissions over the same time-frequency and spreading factor (SF) resources generally result in packet collisions, as the gateway cannot distinguish the overlapping signals from different end devices (EDs). This paper advocates a new fluid antenna multiple access (FAMA) framework for LoRa, referred to as {\it lora}-FAMA, to provide spatial opportunities for LoRa multiuser communications. In {\it lora}-FAMA, a gateway employs a single fluid antenna connected to only one radio-frequency (RF) chain, whose radiating element traverses the antenna aperture by sequentially visiting all candidate positions, i.e., fluid antenna `ports', within each symbol interval. The signal segments collected along the trajectory are compensated using the channel state information for the target ED. As a result, the desired signal is coherently accumulated, whereas the signals from other EDs experience unmatched channel variations and thus cannot be coherently combined. Applying the compensation separately to each active ED enables simultaneous multiuser transmission. We analyze the statistical performance of {\it lora}-FAMA under asynchronous transmissions and spatially correlated fading, as well as the large-aperture limiting case with independent and identically distributed fading. Numerical results show close agreement between the analytical and Monte Carlo results. It is revealed that, with a normalized aperture of $4\times4$ and $\mathrm{SF}=9$, a gateway can simultaneously serve more than $10$ EDs over the same frequency and SF resources while maintaining a symbol error rate below $10^{-4}$. These results demonstrate the potential of fluid antennas to enable LoRa multiple access without multiple RF chains.
\end{abstract}	
\begin{IEEEkeywords}
	Fluid antenna multiple access (FAMA), long-range (LoRa), spatial correlation, symbol error rate (SER).
\end{IEEEkeywords}
\section{Introduction}
Recently, with the rapid development of the Internet-of-Things (IoT), there has been a strong demand for low-power wide-area networks (LPWANs) to support massive connectivity for low-cost, low-complexity end devices (EDs) \cite{bibitem1-1}. Amongst the existing LPWAN technologies, e.g., LTE-M, Sigfox, narrowband (NB)-IoT, and long-range (LoRa), LoRa has been widely adopted owing to its long transmission range, capability of supporting massive connectivity, and flexible network architecture \cite{bibitem1-2}. This can be attributed to its unique chirp-spread-spectrum (CSS) modulation at the physical layer (PHY) and the ALOHA protocol at the medium access control (MAC) layer \cite{bibitem1-3,bibitem1-4}. Specifically, CSS modulation with matched dechirping enables effective energy accumulation over the symbol duration and bandwidth, allowing detection at negative SNRs in decibels (dB), thereby supporting long-distance communications \cite{bibitem1-5}. At the MAC layer, the ALOHA-based random-access protocol enables a gateway to connect thousands of intermittently active and uncoordinated EDs to share a common medium without centralized scheduling \cite{bibitem1-6}. Despite these advantages, both CSS and ALOHA pose challenges to supporting simultaneous transmissions from multiple EDs, i.e., multiple access. The CSS signals with the same spreading factor (SF) and time-frequency resource are superimposed and overlapped during demodulation, making the signals from different EDs indistinguishable \cite{bibitem1-7,bibitem1-8}. Under the ALOHA mechanism, such simultaneous packet arrivals are regarded as a collision, in which only the stronger packet can be decoded due to the capture effect, while the other packets are lost\cite{bibitem1-9,bibitem1-10}. To further enhance connectivity and multiplexity, additional changes are required in the current LoRa networks.
\subsection{Related Works on CSS Modulation}
Enhancing MAC-layer throughput through PHY techniques is a common approach \cite{bibitem1-11}. To enhance the multiuser or multi-packet capability of LoRa, most works focus on new waveform design on the conventional CSS scheme. As summarized in \cite{bibitem1-12,bibitem1-13}, existing LoRa-like waveform designs can be classified into multi-parameter CSS, signal-set extension, chirp multiplexing, index modulation, and hybrid CSS. Specifically, multi-parameter CSS exploits additional chirp parameters to convey information \cite{bibitem1-14,bibitem1-15}, whereas signal-set extension enlarges the set of available chirp waveforms \cite{bibitem1-16,bibitem1-17}. Chirp multiplexing transmits multiple chirps within one symbol \cite{bibitem1-18}, while index modulation maps information onto the indices or activation patterns of transmission entities \cite{bibitem1-19}. Hybrid CSS combines two or more of these mechanisms \cite{bibitem1-20,bibitem1-21}. However, the resulting SE gains are generally accompanied by increased detection complexity, reduced robustness to synchronization impairments, or degraded symbol error ratio (SER) and energy efficiency. Additionally, the achievable SE of these schemes is generally only a few times that of conventional LoRa \cite{bibitem1-13}. More importantly, these SE-oriented waveform designs do not inherently increase the number of users that can be simultaneously supported. Few studies have investigated multiuser transmission in LoRa systems, which rely on complex successive interference cancellation (SIC) receivers \cite{bibitem1-22,bibitem1-23}, as well as differences in SFs \cite {bibitem1-24}, time offsets \cite{bibitem1-8}, and others. Note that existing SE-enhancement and multiuser schemes mainly reorganize finite waveform, time-frequency, and SF resources, making interference and collisions unavoidable as the number of streams or users grows. To further enhance the multiplexing, the spatial domain offers a promising additional degree of freedom (DoF).
\subsection{Related Works on FAMA}
Beyond providing spatial diversity gains, fluid antenna systems (FAS) \cite{bibitem1-25,bibitem1-26}, also known as movable antennas \cite{bibitem1-27,bibitem1-28}, can enable multiuser access through the fluid antenna multiple access (FAMA) technique \cite{bibitem1-29}. The initial FAMA framework, referred to as {\it fast}-FAMA \cite{bibitem1-29,bibitem1-30}, proposed by Wong {\it et al.}, allows the fluid antenna to select the position, i.e., ‘port’, within its aperture that maximizes the instantaneous signal-to-interference ratio (SIR) or signal-to-interference-plus-noise ratio (SINR), thereby mitigating multiuser interference. Then, as a practical alternative, {\it slow}-FAMA performs position switching only when the channel state information (CSI) changes, with reduced multiplexing capability \cite{bibitem1-31}. By selecting favorable users, opportunistic FAMA exploits multiuser diversity to achieve a higher multiplexing gain. \cite{bibitem1-32}. The compact ultra-massive array (CUMA) is another extension of FAMA \cite{bibitem1-33}, which coherently selects and combines the received signals from multiple ports. Note that FAS and FAMA can also be integrated into several emerging techniques, including joint source-channel coding (JSCC) \cite{bibitem1-34}, and other multiplex schemes such as non-orthogonal multiple access (NOMA) \cite{bibitem1-35,bibitem1-36} or fifth-generation (5G) New Radio (NR) \cite{bibitem1-37}. However, existing FAMA architectures are not directly suited to LoRa random access. First, since the port with maximum SINR is generally different for each user, the number of required fluid antennas scales with the number of active EDs, which varies as LoRa EDs are randomly activated. Second, each fluid antenna requires the per-port desired-signal and aggregate-interference information of its associated ED. These limitations motivate the development of a new FAMA framework for LoRa multiple access.
\subsection{Motivation and Contributions}
Motivated by the spatial potential of fluid antennas and the success of FAMA, we believe FAMA offers a promising avenue to enhance the multiuser capability of current LoRa networks. However, existing FAMA frameworks cannot be directly integrated into LoRa networks due to the aforementioned challenges. The opportunity arises from observations of LoRa signals, which show that when the channel varies continuously, correlation-based detection, also known as matched filtering, at the receiver can no longer demodulate the signal correctly. Accordingly, we develop a new multiple-access framework, referred to as {\it lora}-FAMA. Different from the conventional selection-based FAS or FAMA \cite{bibitem1-25,bibitem1-29}, the fluid antenna in {\it lora}-FAMA does not remain at a fixed port but continuously traverses, or `flows' across, the aperture within each symbol interval.\footnote{Unlike FAMA and CUMA, {\it lora}-FAMA can traverse ports along a predetermined path without online optimization. For example, the radiating element may move from one end to the other in the one-dimensional (1D) linear FAS, or sequentially visit the ports of a two-dimensional (2D) aperture.} At the receiver, the symbol segments collected from different ports are compensated using the CSI of each ED. Consequently, the desired-signal energy is coherently accumulated, whereas the interfering signals experience unmatched channel variations and therefore cannot be demodulated. The proposed multiple-access scheme will be detailed in the following sections, and its effectiveness will be demonstrated through statistical performance analysis and Monte Carlo simulations.
The contribution of this paper can be summarized as follows.
\begin{itemize}
\item We propose a FAMA framework for LoRa networks, referred to as {\it lora}-FAMA, that enables simultaneous multiuser transmissions over the same time-frequency and SF resources with a single fluid antenna connected to a single RF chain.

\item We derive statistical characterizations of the LoRa decision variables for asynchronous transmissions under spatially correlated fading with a relatively small antenna aperture and independent and spatially independent fading with a sufficiently large antenna aperture.

\item Since the SER is determined by the relationship between the desired decision value and the maximum undesired decision value, we first characterize the desired decision variable from its statistical moments. The undesired decision variables are generally not independent, their maximum is difficult to characterize directly. Given that the SER is primarily governed by its upper-tail behavior, we establish upper-tail asymptotic independence among different undesired bins and then approximate the maximum using extreme value theory (EVT).

\item To account for dependence among decision variables across different candidate symbols, we develop a copula-based analysis framework and derive a theoretical expression for the system SER.

\item Numerical results validate the accuracy of the  analytical expression and demonstrate the multiuser capability of {\it lora}-FAMA while maintaining a low SER. In particular, with a normalized aperture of $4\times4$ and $\mathrm{SF}=9$, more than $10$ EDs can be simultaneously supported with an SER below $10^{-4}$.
\end{itemize}
\section{Network Model}
\subsection{Practical LoRa Networks}
In LoRa networks, the CSS-based PHY scheme together with the ALOHA-based MAC mechanism forms a basic communication framework. Under CSS modulation, symbols using the same SF share the same spectral characteristics and cannot be distinguished at the receiver. As a result, when multiple symbols, or equivalently multiple packets, are transmitted simultaneously, the receiver cannot reliably separate them, leading to packet collisions\cite{bibitem1-5,bibitem2-1}. To address this and support a large number of IoT EDs with sporadic transmission demands, ALOHA employs a random-access scheme in which EDs wake up only when data is generated and return to sleep after packet delivery \cite{bibitem2-2}. This process relies on a sufficiently low duty cycle to maintain stable network operation \cite{bibitem2-3,bibitem2-4}. Nevertheless, as the number of EDs grows, the aggregate duty cycle increases, leading to more severe packet collisions. Moreover, because EDs are randomly activated and transmit without coordination, the corresponding signals generally arrive at the gateway asynchronously. In practice, the gateway detects the arrival of each packet using the preamble and then demodulates the payload directly, without additional synchronization or handshaking between the ED and the gateway.  

Finally, uplink typically dominates in practical LoRa networks due to the asymmetric nature of IoT services. Downlink transmissions are primarily used for ED parameter configuration and control and are therefore triggered only at specific time instants. To support downlink reception, LoRa networks define three operating modes, namely Class A, B, and C \cite{bibitem2-4,bibitem2-5}. In Class A, which is the most widely adopted mode, each ED opens two short receive windows immediately after an uplink transmission. In Class B, the ED periodically opens additional receive windows in accordance with gateway-scheduled ping slots. In Class C, the ED remains almost continuously in the receive state and performs persistent preamble detection. Due to stringent power constraints on the ED side, most practical deployments operate in Class A. Finally, despite the availability of the downlink classes above, LoRa EDs typically operate with unconfirmed uplink transmissions. Therefore, packet collisions generally result in the loss of the corresponding data. Besides, it is generally difficult to establish synchronization between EDs and the gateway through the downlink link.

\begin{table}[t]
	\centering
	\caption{Summary of key notations}
	\begin{tabular}{m{1.3cm}<{\centering}||l}
		\hline
		Notation & \makecell[c]{Meaning} \\ \hline
		\multirow{2}{*}{$K$} 
		& Number of transmit-symbol indices, discrete-time \\ 
		& sample points, and DFT bins \\ \hline
		$N_0$ 
		& Single side noise power spectral density \\ \hline
		\multirow{2}{*}{$N_1, N_2$} 
		& Numbers of FAS ports along the horizontal\\
		& and vertical dimensions \\ \hline
		$N$ 
		& Total number of FAS ports \\ \hline
		\multirow{2}{*}{$W_1, W_2$} 
		& Normalized FAS dimensions along the horizontal\\
		& and vertical directions \\ \hline
		$W$ 
		& FAS aperture \\ \hline
		$a_{j}$ 
		& Transmit symbol index of the $j$-th ED \\ \hline
		$n$ 
		& Index of the discrete-time sample point \\ \hline
		$l$ 
		& Index of the FAS port \\ \hline
		$k$ 
		& Index of the DFT bin \\ \hline
		$\tau$ 
		& Arrival-time offset of the reference ED \\ \hline
		${d}_{j}$ 
		& Relative timing offset of the $j$-th interfering ED \\ \hline
		\multirow{2}{*}{$\mathcal{K}$} 
		& Set of transmit-symbol indices, discrete-time sample\\ 
		& points, and DFT bins \\ \hline
		$\mathcal{U}$ 
		& Set of ED indices \\ \hline
		$\mathcal{L}$ 
		& Set of FAS port indices \\ \hline
		$\mathcal{S}_{l\mid\tau}$
		& Set of sample indices received through $l$-th port \\ \hline
		\multirow{2}{*}{$\mathcal D_{j,\mathrm{p}}, \mathcal D_{j,\mathrm{q}}$} 
		& Sets of sample indices corresponding to the tail \\
		& and head fragments of the $j$-th interfering ED\\ \hline
	\end{tabular}
	\label{table_key_variables}
\end{table}

\begin{remark}
Based on the above discussion, we aim to design a multiple-access scheme for LoRa uplink transmission to scale up ED deployments. In this work, we do not model the random activation behavior of EDs, as the MAC layer governs it. Instead, we consider the PHY scenario in which multiple EDs are already active and transmit simultaneously to the gateway. We consider randomly arriving packets, with no explicit synchronization between the EDs and the gateway. Since the LoRa transmission rate is fixed for a given parameter configuration, the multiplexing capability of the proposed scheme is not quantified by the multiplexing gain. Instead, we evaluate it based on SER. In particular, a set of simultaneously active EDs is regarded as supportable if the SER achieved by each ED is below a target threshold, e.g., $10^{-3}$ or $10^{-4}$.
\end{remark}
\subsection{{\it lora}-FAMA Scheme}
We consider an uplink LoRa network in which a gateway simultaneously receives transmissions from $U$ active EDs. In conventional LoRa systems, concurrent transmissions often result in severe packet collisions, preventing the gateway from separating the superposed signals and leading to decoding failure. To end this, we propose a {\it lora}-FAMA scheme in which the gateway employs a single FAS and a single RF chain. The proposed receiver requires no modification at the EDs. The main notation is summarized in Table~\ref{table_key_variables}.

\begin{figure}[!t]
	\centering
	\includegraphics[width=1\linewidth]{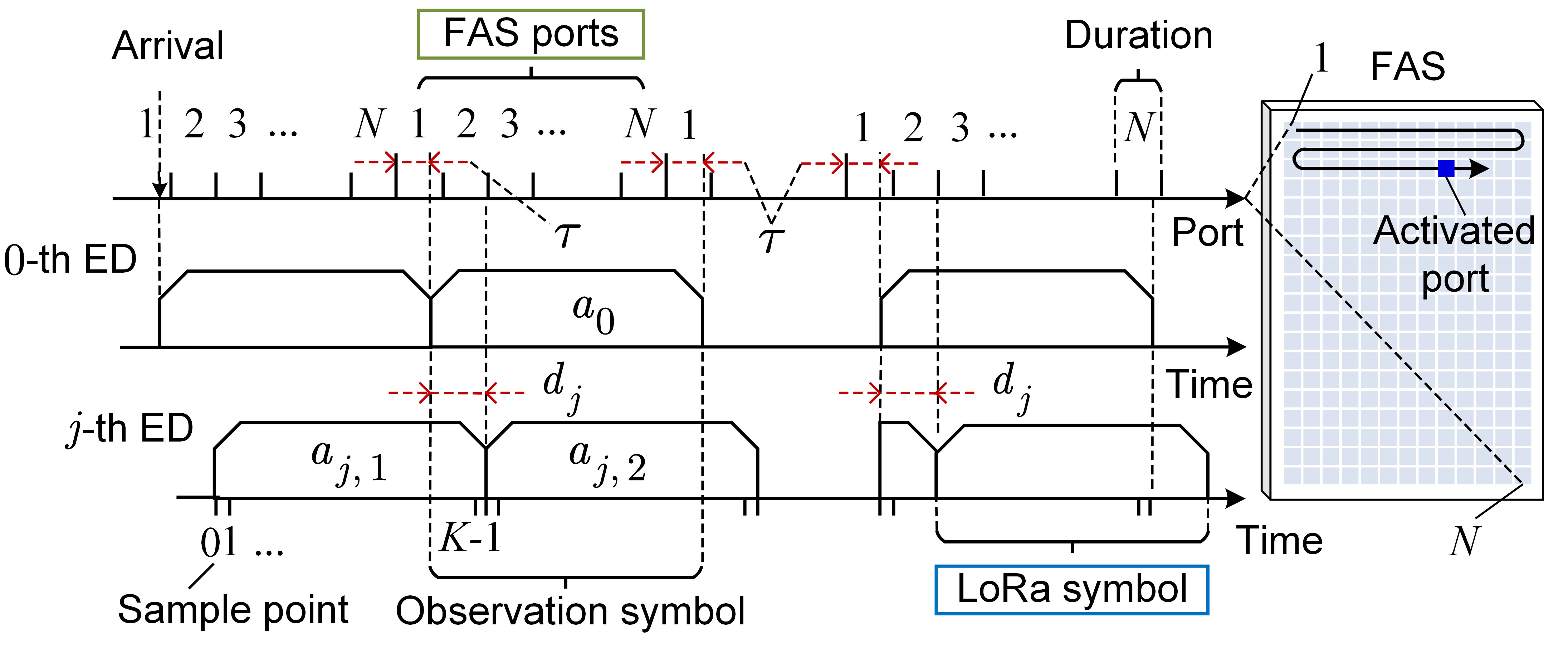}%
	\caption{Schematic diagram of the proposed symbol structure.}
	\label{symbol}
\end{figure}

For $SF\in\{7,8,\ldots,12\}$, the LoRa alphabet contains $K=2^{SF}$ symbols. With the standard sampling rate $f_{\mathrm{s}}=B$, each LoRa symbol is represented by $K$ discrete-time samples over one symbol interval $T_{\mathrm{s}}=K/B$. Accordingly, the discrete-time signal transmitted by the $j$-th ED is \cite{bibitem2-7}
\begin{equation}\label{xm}
	x_{a_{j}}[n]
	=
	\sqrt{\frac{1}{K}}
	\exp\!\left[
	j2\pi\left(
	\frac{n^2}{2K}
	+
	\left(\frac{a_{j}}{K}-\frac{1}{2}\right)n
	\right)
	\right],\:n\in\mathcal K,
\end{equation}
where $\mathcal K\triangleq\{0,1,\ldots,K-1\}$, $j
\in\mathcal U\triangleq\{0,1,\ldots,U-1\}$, and $a_{j}\in\mathcal K$ denotes the transmitted symbol index. The basic LoRa signal corresponds to $a_{j}=0$.

The gateway is equipped with a 2D FAS of physical size $W_1\lambda\times W_2\lambda$, where $\lambda$ denotes the carrier wavelength.\footnote{For the considered scenario, a CUMA-based 2D FAS implemented with MEMS pixel antennas is more suitable than fluid material-based FAS implementations, since electronically switchable antenna pixels can provide fast position switching without rapid physical movement of the radiating element \cite{bibitem1-33}.} The FAS aperture is uniformly discretized into $N=N_1N_2$ ports. Let $l\in\mathcal L\triangleq\{1,2,\ldots,N\}$ denote the one-dimensional port index. The corresponding horizontal and vertical indices are given by \cite{bibitem1-33}
\begin{equation}
	n_1 = ((l-1)\bmod N_1)+1,\:
	n_2 = \left\lfloor\frac{l-1}{N_1}\right\rfloor+1.
\end{equation}

Under rich scattering, the channel coefficient between the $j$-th ED and port $l$ is denoted by $h_{j}[l]$ and modeled as $h_{j}[l]\sim\mathcal{CN}(0,\Omega)$, where $\Omega=\mathbb{E}[|h_{j}[l]|^2]$. For two ports $u=\mathrm{map}(n_1,n_2)$ and $v=\mathrm{map}(n_3,n_4)$, the spatial correlation is given by\cite{bibitem2-8}
\begin{equation}\label{jake-model}
	\begin{aligned}
		&\mathbb{E}\!\left[h_{j}[u]\bigl(h_{j}[v]\bigr)^*\right]=\Omega\left[\mathbf{R}\right]_{u,v} \\
		&\quad =
		\Omega J_0\!\left(
		2\pi
		\sqrt{
			\left(\frac{(n_1-n_3)W_1}{N_1-1}\right)^2
			+
			\left(\frac{(n_2-n_4)W_2}{N_2-1}\right)^2
		}
		\right),
	\end{aligned}
\end{equation}
where $J_0(\cdot)$ denotes the zeroth-order Bessel function of the first kind and $\mathbf{R}$ is the correlation matrix of ports. Let $\mathbf{h}_{j}=\left[h_{j}[1],h_{j}[2],\cdots,h_{j}[N]\right]^{\mathrm T}$, we have $\mathbf h_{j}\sim\mathcal{CN}(\mathbf{0},\Omega\mathbf R)$.

In contrast to conventional {\it fast}-FAMA and {\it slow}-FAMA, where the RF chain remains on a selected port during each symbol interval, the proposed {\it lora}-FAMA scheme sequentially traverses all $N$ ports within one LoRa symbol, as illustrated in Fig.~\ref{symbol}. This operation is feasible because the LoRa symbol duration is substantially longer than the switching time of typical MEMS RF switches.\footnote{In LoRa deployments, $B=125~\mathrm{kHz}$ with $SF=7$ to $12$ yields a symbol duration from $1.024~\mathrm{ms}$ to $32.768~\mathrm{ms}$, while the switching time of MEMS RF switches is commonly ranging from a few to several tens of microseconds ($\mu{s}$)\cite{bibitem2-9}.} For analytical tractability, we assume that $N$ divides $K$ and define $Q\triangleq K/N$. Then, the effective channel of the $j$-th ED at sample $n$ is
\begin{equation}\label{group-h}
	g_{j}[n] = h_{j}[l],\quad (l-1)Q \le n \le lQ-1.
\end{equation}

Without loss of generality, for the reference ED, the port corresponding to its symbol-arrival instant is indexed by $l=1$. Since the desired signal may arrive at any instant while the RF chain dwells on a given port, we denote the corresponding timing offset by $\tau$. For tractability, only integer-valued offsets are considered, and ${\tau}$ is assumed to be uniformly distributed over $\{0,1,\ldots,Q-1\}$. Given $\tau$, the resulting channel sequence over one observation window is
\begin{equation}
	g_{j\mid\tau}[n] = g_{j}[(n+\tau)\bmod K],\: n\in\mathcal K.
\end{equation}
To explicitly characterize the port support within the current observation
window, define the set of sample indices received through port $l$ as $\mathcal S_{l\mid\tau}$, which is given by
\begin{subequations}
	\begin{align}
		\mathcal S_{1\mid\tau}
		&=
		\{0,\cdots,Q-\tau-1\}\cup\{K-\tau,\cdots,K-1\},\\
		\mathcal S_{l\mid\tau}
		&=
		\{(l-1)Q-\tau,\cdots,lQ-\tau-1\},\: l\neq 1,
	\end{align}
\end{subequations}
where $\tau\in\{0,1,\cdots,Q_1\}$ and $Q_1=Q-1$. Note that only the
first-port is split into two disjoint segments.

For notation simplicity, the $0$-th ED is selected as the reference.\footnote{The same detection procedure is performed for every active ED, with each ED can be treated as the reference.} Accordingly, the desired-signal component is
\begin{subequations}
	\begin{align}
		r_{0}[n]&= g_{0\mid\tau}[n]\sqrt{E_{\mathrm{s}}}\,x_{a_{0}}[n],\:n\in\mathcal K,\\
		&= h_{0}[l]\sqrt{E_{\mathrm{s}}}\,x_{a_{0}}[n],\:n\in\mathcal{S}_{\l\mid\tau},
	\end{align}
\end{subequations}
where $E_{\mathrm{s}}$ denotes the symbol energy, and $a_{0}$ is the transmit symbol for $0$-th ED.

For the $j$-th interfering ED. Let ${d}_{j}$ denote the integer-valued timing offset relative to the reference observation window, where ${d}_{j}$ is uniformly distributed over $\mathcal K$. Over one symbol window, the observed interference generally contains the tail of one LoRa symbol and the head of the next. Let $a_{j,\mathrm{p}}$ and $a_{j,\mathrm{q}}$ denote the corresponding symbol indices. Besides, denote $
\boldsymbol{\upsilon}_{j}=\big(a_{j,\mathrm{p}},a_{j,\mathrm{q}},{d}_{j}\big)\in\mathcal{K}^3$. All discrete variables are independently and uniformly distributed over their domains. Then, the interfering waveform is expressed as
\begin{equation}\label{xj_piece}
	\begin{aligned}
		x_{\boldsymbol{\upsilon}_{j}}[n]
		=
		\begin{cases}
			x_{a_{j,\mathrm{p}}}[K-{{d}_{j}}+n], & n\in\mathcal D_{j,\mathrm{p}},\\
			x_{a_{j,\mathrm{q}}}[n-{{d}_{j}}], & n\in\mathcal D_{j,\mathrm{q}}.
		\end{cases}
	\end{aligned}
\end{equation}
where $\mathcal{D}_{j,\mathrm{p}}=\big\{0,1,\ldots,{d}_{j}-1\big\}$ and $\mathcal{D}_{j,\mathrm{q}}=\big\{{d}_{j},{d}_{j}+1,\ldots,K-1\big\}$. Here, the subscripts `$\mathrm{p}$' and `$\mathrm{q}$' denote the \emph{tail} and \emph{head} fragments observed in the current observation window, respectively. Hence, the overall received signal at the gateway is given by\footnote{Since the observation window is only defined by the arrival instant of the reference ED and its corresponding FAS port, the same timing offset $\tau$ applies to the other EDs as well.}
\begin{equation}
	r[n]=r_{0}[n]+\sum_{\substack{j=1}}^{U_1}
	g_{j\mid\tau}[n]\sqrt{E_{\mathrm{s}}}\,x_{\boldsymbol{\upsilon}_{j}}[n]+z[n],
\end{equation}
where $U_1=U-1$ and the additive noise is modeled as $z[n]\sim\mathcal{CN}(0,N_0)$, where $N_0$ is the noise power spectral density. Since $T_{\mathrm{s}}=K/B$, the received SNR can be written as\cite{bibitem2-10}
\begin{equation}
	\Gamma\triangleq \frac{E_{\mathrm{s}}}{T_{\mathrm{s}}N_0B}
	= \frac{E_{\mathrm{s}}}{N_0K}.
\end{equation}
For simplicity, we set $E_\mathrm{s}=1$, such that different SNR values are obtained by varying $N_0$.

Assuming that full CSI is available at the gateway, coherent detection for the $0$-th ED is performed by normalizing, channel-conjugate combining, dechirping with the basic LoRa signal, and a $K$-point DFT. Specifically, given $\tau$, the DFT output is obtained as
\begin{subequations}\label{wk_norm}
	\begin{align}
		{w}_{0}[k]
		&=
		\operatorname{DFT}\left\{
		{\bigl(g_{0\mid\tau}[n]\bigr)^*}
		r[n]
		\bigl(x_0[n]\bigr)^*
		\right\} \\
		&=
		\sum_{n=0}^{K-1}
		\bigl(g_{0\mid\tau}[n]\bigr)^*
		r[n]
		\bigl(x_0[n]\bigr)^*
		\exp\!\left(-j2\pi\frac{kn}{K}\right) \\
		&=
		\sum_{n=0}^{K_1}
		\bigl(g_{0\mid\tau}[n]\bigr)^*
		r[n]\bigl(x_k[n]\bigr)^*,\:k\in\mathcal{K},
	\end{align}
\end{subequations}
where $K_1=K-1$ and $k$ is known as the frequency bin. Moreover, $x_k[n]$ follows from \eqref{xm} and $k$ is known as the index of the DFT bin. The real-valued decision metric is then utilized. Finally, the detected symbol is obtained as $\hat{a}_{0} = \arg\max\limits_{k\in\mathcal K} \Re\left\{w_{0}[k]\right\}$. A symbol error occurs when $\hat{a}_{0}\neq a_{0}$. Therefore, the SER is

\begin{subequations}\label{ser}
\begin{align}
	P_{\mathrm{s}}
	&=\Pr\Big(\max_{\substack{k\in\mathcal K,\:k\neq{a}_{0}}}
	\Re\left\{w_{0}[k]\right\}>\Re\left\{w_{0}\big[a_{0}\big]\right\}\Big)\\
	&\triangleq\Pr\Big(\max_{\substack{k\in\mathcal K,\:k\neq a_{0}}}\varpi_{0}[k]>\varpi_{0}\big[a_{0}\big]\Big)
\end{align}
\end{subequations}

The proposed {\it lora}-FAMA scheme constructs a time-varying spatial matched filter with only a single RF chain, where the power of interference is dispersed over the DFT bins after coherent processing to avoid generating a correlation peak. When the RF chain remains connected to a single port throughout the entire symbol interval, the proposed receiver reduces to a conventional LoRa receiver, in which interference produces concentrated peaks in the DFT output.

\section{Performance Analysis}
In this section, we aim to evaluate the performance of the proposed {\it lora}-FAMA network. For tractability, we first decompose the real part of the DFT output in \eqref{wk_norm} as
\begin{equation}\label{w-2}
	\begin{split}
		\varpi_{0}[k]=\Re&\Bigg\{\frac{\sqrt{E_\mathrm{s}}}{K}\sum_{n=0}^{K_1}
		G_{0,0\mid\tau}[n]\exp\left(j2\pi\frac{(a_{0}-k)n}{K}\right)\\
		&+\sqrt{E_\mathrm{s}}\sum_{j=1}^{U_1}\sum_{n=0}^{K_1}G_{0,j\mid\tau}[n]x_{\boldsymbol{\upsilon}_{j}}[n]\big(x_k[n]\big)^*\\
		&+\sum_{n=0}^{K_1}\big(g_{0\mid\tau}[n]\big)^*z[n]\big(x_k[n]\big)^*\Bigg\}\\
		=\varpi_{\mathrm{s},0}&[k]+\underbrace{\sum_{j=1}^{U_1}\varpi_{\mathrm{I},0,j}[k]}_{\varpi_{\mathrm{I},0}[k]}
		+\varpi_{\mathrm{n},0}[k].
	\end{split}
\end{equation}
and $G_{0,j\mid\tau}[n]=\big(g_{0\mid\tau}[n]\big)^*g_{j\mid\tau}[n]$ is the product of two complex Gaussian channel coefficients. Then, the exact cumulative distribution function (CDF) of decision values at the $k$-th bin can be obtained as 
\begin{figure*}[!t]
	\normalsize
	\begin{equation}\label{cdf-exc}
		\begin{split}
			F_{\varpi_{0}[k]}(x)={}&\frac{1}{QK^{3U_{1}}}\sum_{\tau=0}^{Q_1}\sum_{\boldsymbol{\upsilon}_{1}\in\mathcal{K}^{3}}\cdots\sum_{\boldsymbol{\upsilon}_{U_{1}}\in\mathcal{K}^{3}}\int_{\mathbb{C}^{N}}\Phi\Bigg(\frac{x-\sqrt{E_{\mathrm{s}}}\mathbf{h}^{\mathrm{H}}\mathbf{C}_{k\mid\tau}\mathbf{h}}{\sqrt{\frac{1}{2}\mathbf{h}^{\mathrm{H}}\big(\frac{N_{0}}{N}\mathbf{I}_{N}+E_{\mathrm{s}}\Omega\sum_{j=1}^{U_{1}}\mathbf{D}_{k\mid\tau,\boldsymbol{\upsilon}_{j}}\mathbf{R}\mathbf{D}_{k\mid\tau,\boldsymbol{\upsilon}_{j}}^{\mathrm{H}}\big)\mathbf{h}}}\Bigg)f_{\mathbf{h}_0}(\mathbf{h})\mathrm{d}\mathbf{h},
		\end{split}
	\end{equation}
	\hrule
\end{figure*}
where $\mathbf{C}_{k\mid\tau}=\mathrm{diag}\big\{\sum_{n\in\mathcal{S}_{l\mid\tau}}\cos\left({2\pi(a_{0}-k)n}/{K}\right)/K\big\}_{l=1}^{N}$, $\mathbf{D}_{k\mid\tau,\boldsymbol{\upsilon}_{j}}=\mathrm{diag}\big\{\sum_{n\in\mathcal{S}_{l\mid\tau}}x_{\boldsymbol{\upsilon}_{j}}[n]\big(x_{k}[n]\big)^{*}\big\}_{l=1}^{N}$ and $f_{\mathbf{h_0}}(\mathbf{h})={\pi^{-N}{\exp\left(-\mathbf{h}^{\mathrm{H}}(\Omega\mathbf{R})^{-1}\mathbf{h}\right)}/\det(\Omega\mathbf{R})}$.

\begin{remark}
Since terms in \eqref{w-2} involve the common $g_{0}[n]$, which are correlated across different $n$ and also include a large number of combinations of discrete variables, the exact CDF in \eqref{cdf-exc} requires a $QK^{3U_1}$-term summation and a $N$-fold integral over the complex plane. Hence, direct numerical evaluation is computationally intensive. Therefore, we first derive approximate marginal distributions for the desired and undesired DFT bins, i.e., $\varpi_{0}[a_{0}]$ and $\varpi_{0}[k]$, $k\neq{a}_{0}$, with reduced computational complexity, and then employ a copula function to characterize the dependence between these two terms. Finally, the SER is obtained from the copula function. 
\end{remark}

\subsection{Desired-Bin Distribution}
An analytically tractable characterization of $\varpi_{0}[a_{0}]$ is obtained by approximating it with a Gamma-distributed random variable. The corresponding parameters are determined by matching its moments, which are derived next.
\begin{lemma}\label{l1}
The mean and variance of the desired-bin decision variable $\varpi_{0}[a_{0}]$ are given by
\begin{equation}\label{e_wi}
	\begin{split}
		\mathrm{E}\left[\varpi_{0}[a_{0}]\right]=\sqrt{E_{\mathrm{s}}}\Omega,
	\end{split}
\end{equation}
\begin{equation}\label{var_wi}
	\begin{split}
		\mathrm{Var}\left[\varpi_{0}[a_{0}]\right]=\frac{E_{\mathrm{s}}\Omega^{2}}{N^{2}}\mathrm{tr}\left(\mathbf{R}^{2}\right)+\frac{\Omega N_{0}}{2}+\frac{E_{\mathrm{s}}\Omega^{2}U_{1}}{2K},
	\end{split}
\end{equation}
respectively.
\end{lemma}

\begin{proof}
See Appendix \ref{app:a}.
\end{proof}

\begin{remark}
Based on \eqref{e_wi} and \eqref{var_wi}, $\varpi_{0}\big[a_{0}\big]$ can be approximated as a Gamma random variable, $\varpi_{0}[a_{0}]\sim\mathrm{Gamma}(\alpha_0,\beta_0)$, where
\begin{equation}
	\alpha_0=\frac{\big(\mathrm{E}\left[\varpi_{0}[a_{0}]\right]\big)^2}{\mathrm{Var}\left[\varpi_{0}[a_{0}]\right]}
	\quad\mathrm{and}\quad
	\beta_0=\frac{\mathrm{Var}\left[\varpi_{0}\big[a_{0}\big]\right]}{\mathrm{E}\left[\varpi_{0}\big[a_{0}\big]\right]}.
\end{equation}
Note that for some special cases, the exact distribution can be obtained as follows.
\end{remark}

\begin{lemma}\label{l2}
For $N=K$ and sufficiently large $W$ such that $\mathbf{R}\to\mathbf{I}_{K}$, the CDF of $\varpi_{0}[a_{0}]$ is given by
\begin{equation}\label{a_i-x>0}
	\begin{split}
		F_{\varpi_{0}[a_{0}]}(x)=1-\exp\left(-\frac{x}{p}\right)\sum_{\kappa=0}^{K_1}\sum_{b=0}^{\kappa}A_{\kappa,b}\frac{x^{\kappa-b}}{p^{\kappa}q^{K}},
	\end{split}
\end{equation}
for $x\geq0$, whereas
\begin{equation}\label{a_i-x<0}
	\begin{split}
		F_{\varpi_{0}[a_{0}]}(x)=\exp\left(\frac{x}{q}\right)\sum_{\kappa=0}^{K_1}\sum_{b=0}^{\kappa}A_{\kappa,b}\frac{(-x)^{\kappa-b}}{p^{K}q^{\kappa}},
	\end{split}
\end{equation}
for $x<0$, where $A_{\kappa,b}=\binom{\kappa}{b}\Gamma(K+b)/\left(\kappa!\Gamma(K)\rho^{K+b}\right)$, $\rho=1/p+1/q$,
$p=(\sqrt{(\theta\Omega)^{2}+2\Omega v}+\theta\Omega)/2$, $q=(\sqrt{(\theta\Omega)^{2}+2\Omega v}-\theta\Omega)/2$, $\theta=\sqrt{E_{\mathrm{s}}}/K$, and $v=N_{0}/(2K)+E_{\mathrm{s}}\Omega U_{1}/(2K^{2})$.
\end{lemma}

\begin{proof}
See Appendix \ref{app:b}.
\end{proof}

\subsection{Undesired-Bin Distribution}

For the undesired bins, where $k\neq{a_{0}}$, the decision metric is generated by the maximum of $K_1$ bins. Therefore, a Gamma approximation based on low-order moments is inadequate for characterizing this extreme value, as the SER is dominated by its upper-tail behavior. In addition, for analytical tractability, we retain the interference and noise components, which dominate the undesired bins, and neglect the comparatively weak leakage of the desired signal.

Here, we first derive the CF for the values in the undesired bins. Then, the undesired bins are proved to be asymptotically independent in the upper tail. Based on this, the maximum value across the bins can be derived using EVT, with two specific quantiles determined by the initial CF.

\begin{remark}
This avoids both the $K_1$-th power operation involved in the extreme-value distribution and the inversion of the distribution over the entire range of the decision variable. 
\end{remark}

Let $X_{k}=\varpi_{\mathrm{I},0}[k]+\varpi_{\mathrm{n},0}[k]$ for $k\neq{a}_{0}$.

\begin{lemma}\label{l3}
The CF of $X_{k}$ can be approximated as
\begin{equation}\label{cf-X}
	\begin{split}
		\Psi_{X_{k}}&(t)\approx\frac{1}{Q}\det\left(\mathbf{M}_{\mathrm{n}}(t)\right)^{-1}\\
		&\sum_{\tau=0}^{Q_{1}}\Bigg[\frac{1}{K^{3}}\sum_{\boldsymbol{\upsilon}\in\mathcal{K}^{3}}
		\det\left(\mathbf{I}_{N}+\frac{t^{2}}{4}\widetilde{\mathbf{B}}_{k\mid\tau,\boldsymbol{\upsilon}}(t)\right)^{-1}\Bigg]^{U_{1}},
	\end{split}
\end{equation}
where $\widetilde{\mathbf{B}}_{k\mid\tau,\boldsymbol{\upsilon}}(t)=\mathbf{M}_{\mathrm{n}}^{-\frac{1}{2}}(t)\mathbf{B}_{k\mid\tau,\boldsymbol{\upsilon}}\mathbf{M}_{\mathrm{n}}^{-\frac{1}{2}}(t)$, $\mathbf{M}_{\mathrm{n}}(t)=\mathbf{I}_{N}+t^{2}N_{0}\Omega\mathbf{R}/(4N)$, and $\mathbf{B}_{k\mid\tau,\boldsymbol{\upsilon}}=E_{\mathrm{s}}\Omega^{2}\mathbf{R}^{\frac{1}{2}}\mathbf{D}_{k\mid\tau,\boldsymbol{\upsilon}}\mathbf{R}\mathbf{D}_{k\mid\tau,\boldsymbol{\upsilon}}^{\mathrm{H}}\mathbf{R}^{\frac{1}{2}}$.
Moreover, $\boldsymbol{\upsilon}=(a_{\mathrm{p}},a_{\mathrm{q}},d)\in\mathcal{K}^{3}$ denotes an arbitrary combination of the two symbol indices and the relative timing offset of an interfering ED.
\end{lemma}

\begin{proof}
See Appendix \ref{app:c}.
\end{proof}

\begin{lemma}\label{l4}
The undesired bin decision variables are asymptotically independent in the upper tail.
\end{lemma}

\begin{proof}
See Appendix \ref{app:d}.
\end{proof}

Since the low-SER regime, e.g., $10^{-4}$, is of primary interest, we focus on the upper-tail behavior of the undesired bins. Based on Lemma \ref{l4}, EVT can be applied to approximate the upper-tail distribution of $X^\star=\max_{k\in\mathcal{K},k\neq{a}_{0}}X_{k}$ and it follows $\mathrm{Gumbel}(b_{K_{1}},a_{K_{1}})$, where 
\begin{equation}
	\begin{split}
		b_{K_{1}}=F_{X_k}^{-1}\left(1-\frac{1}{K_{1}}\right)\:\mathrm{and}\:a_{K_{1}}=F_{X_k}^{-1}\left(1-\frac{1}{eK_{1}}\right)\!-\!b_{K_{1}}.
	\end{split}
\end{equation}
The two quantiles can be obtained by   using the Gil-Pelaez inversion on \eqref{cf-X} and solved numerically.

\begin{corollary}\label{c1}
For $N=K$ and sufficiently large $W$ such that $\mathbf{R}\to\mathbf{I}_{K}$, the CDF of $X_{k}$ is given by directly setting $\theta=0$ in Lemma 2, noted as $F_{X_k}(x)$.
\end{corollary}

\begin{corollary}\label{c2}
Under the same condition with Corollary~\ref{c1}, $X_{k}$ can be approximated as a Gaussian random variable with the variance  $\sigma_{X}^{2}=\Omega N_{0}/2+E_{\mathrm{s}}\Omega^{2}U_{1}/(2K)$. This follows from \eqref{w-cf}. By setting $\theta=0$ and defining $\sigma_{X}^{2}=K\Omega v$, we have
\begin{equation}
	\begin{split}
		\Psi_{X_{k}}(t)&=\left(1+\frac{\Omega vt^{2}}{2}\right)^{\!-K}
		=\left(1+\frac{\sigma_{X}^{2}t^{2}}{2K}\right)^{\!-K}\\
		&\approx\exp\left(-\frac{\sigma_{X}^{2}t^{2}}{2}\right),
	\end{split}
\end{equation}
for any ${k}\in{\mathcal{K}}\setminus\{a_{0}\}$, where the last expression is the CF of a zero-mean Gaussian RV with variance $\sigma_{X}^{2}$.
\end{corollary}

\begin{algorithm}[t]
	\caption{Numerical Evaluation of $\delta$}
	\label{alg:delta}
	\begin{algorithmic}[1]
		\REQUIRE $K$, $N$, $U_{1}$, $E_{\mathrm{s}}$, $N_{0}$, $\Omega$, and $\mathbf{R}$
		\PARAMETER $L_{0}$ is the number of base deterministic integration points.
		\ENSURE $\hat{\delta}$
		\STATE Set
		$Q=K/N$,
		$D={2NU}+{2K}+{1}+{3U_{1}}$.
		\STATE Construct $\{\mathbf{u}_{s}\}_{s=1}^{L_{0}}\subset(0,1)^{D}$
		using a $D$-dimensional Kronecker sequence generated from the square roots of the first $D$ prime numbers. Then, set
		$\mathbf{u}_{L_{0}+s}=1-\mathbf{u}_{s}$ for $s=1,\cdots,L_0$.
		\FOR{$s=1,\ldots,2L_0$}
		\STATE Transform the first $2N(U_{1}+1)$ coordinates of $\mathbf{u}_{s}$ into the desired and interfering channels using the inverse of standard normal CDF $\Phi^{-1}(\cdot)$ and $(\Omega\mathbf{R})^{1/2}/\sqrt{2}$.
		\STATE Map the next coordinate uniformly to $\tau^{(s)}\in\{0,\ldots,Q-1\}$.
		\STATE Map the following $3U_{1}$ coordinates uniformly to $\{\boldsymbol{\upsilon}_{j}^{(s)}\}_{j=1}^{U_{1}}$ over their  possible values.
		\STATE Transform the remaining $2K$ coordinates into the complex noise samples using $\Phi^{-1}(\cdot)$ and $\sqrt{N_{0}/2}$.
		\STATE Evaluate $\varpi_{0}^{(s)}[a_{0}]$ and
		$X_{k}^{(s)}
		=
		\varpi_{\mathrm{I},0}^{(s)}[k]
		+
		\varpi_{\mathrm{n},0}^{(s)}[k]$
		from \eqref{w-2} and set
		$X_{s}^{\star}
		=
		\max_{k\in\mathcal{K},\,k\neq a_{0}}
		X_{k}^{(s)}$.
		\ENDFOR
		\STATE Compute
		\begin{equation*}
			\begin{split}
				&\hat{\tau}_{\mathrm{K}}^{(L)}
				={}
				\frac{2}{2L_0(2L_0-1)}
				\\
				&\sum_{1\leq r<s\leq L}
				\operatorname{sgn}\Big[
				\big(
				\varpi_{0}^{(r)}[a_{0}]
				-
				\varpi_{0}^{(s)}[a_{0}]
				\big)
				\big(
				X_{r}^{\star}
				-
				X_{s}^{\star}
				\big)
				\Big].
			\end{split}
		\end{equation*}
		\STATE Set
		$\hat{\delta}
		=
		\big(
		1-\hat{\tau}_{\mathrm{K}}^{(L)}
		\big)^{-1}$.
		\STATE \textbf{return} $\hat{\delta}$.
	\end{algorithmic}
\end{algorithm}

\subsection{SER Performance}

Due to the common discrete variables and shared channel, $\max_{k\in\mathcal{K},k\neq{a}_{0}}X_{k}$ is not independent with $\varpi_{0}[a_{0}]$. However, the dependence is difficult to analytically represent due to the nonlinear maximum operation and the loss of orthogonality among the segmented LoRa symbols.  

According to Sklar’s theorem, the joint CDF of $\varpi_{0}[a_0]$ and $X^\star$ under the Gumbel copula,\footnote{The Gumbel Copula is employed since it exhibits positive upper-tail dependence and is therefore suitable for characterizing the tendency of $\varpi_{0}[a_{0}]$ and $X^{\star}$ to take large values simultaneously.} can be expressed as \cite{bibitem3-1,bibitem3-2}
\begin{equation}
	\begin{split}
		F_{\varpi_{0}[a_0],X^{\star}}(y,x)=C_{\delta}\big(F_{\varpi_{0}[a_0]}(y),F_{X^\star}(x)\big),
	\end{split}
\end{equation}
where $C_{\delta}(u,v)=\exp\big(-\left[(-\log u)^{\delta}+(-\log v)^{\delta}\right]^{1/\delta}\big)$ and the dependence parameter is given by $\delta\geq1$. Finally, the analytical SER is given by the following Lemma.

\begin{lemma}\label{l5}
The SER of the proposed {\it lora}-FAMA scheme can be approximated by \eqref{ser-1} for a relatively small antenna aperture and by \eqref{ser-2} for $N=K$ and $\mathbf{R}=\mathbf{I}_{K}$,
\begin{figure*}[!ht]
	\normalsize
	\begin{equation}\label{ser-1}
		\begin{split}
			P_{\mathrm{s}}\approx\int_{0}^{\infty}&\left[1-C_{\delta}^{(1)}\left(\frac{\gamma\left(\alpha_{0},x/\beta_{0}\right)}{\Gamma(\alpha_{0})},\exp\left[-\exp\left(-\frac{x-b_{K_{1}}}{a_{K_{1}}}\right)\right]\right)\right]\frac{x^{\alpha_{0}-1}}{\Gamma(\alpha_{0})\beta_{0}^{\alpha_{0}}}\exp\left(-\frac{x}{\beta_{0}}\right)\mathrm{d}x.
		\end{split}
	\end{equation}
	\hrule
\end{figure*}
\begin{figure*}
	\begin{equation}\label{ser-2}
	\begin{split}
		P_{\mathrm{s}}\approx\int_{-\infty}^{\infty}
		&\left[
		1-C_{\delta}^{(1)}
		\left(
		F_{\varpi_{0}[a_{0}]}(x),
		\left[
		\Phi\left(\frac{x}{\sigma_{X}}\right)
		\right]^{K_{1}}
		\right)
		\right]
		\begin{cases}
			\displaystyle
			\frac{\exp\left(x/q\right)}
			{\Gamma(K)^{2}p^{K}q^{K}}
			\sum_{b=0}^{K_{1}}
			\binom{K_{1}}{b}
			\frac{\Gamma(K+b)(-x)^{K_{1}-b}}
			{\rho^{K+b}},
			&x<0,\\[3mm]
			\displaystyle
			\frac{\exp\left(-x/p\right)}
			{\Gamma(K)^{2}p^{K}q^{K}}
			\sum_{b=0}^{K_{1}}
			\binom{K_{1}}{b}
			\frac{\Gamma(K+b)x^{K_{1}-b}}
			{\rho^{K+b}},
			&x\geq0
		\end{cases}
		\mathrm{d}x.
	\end{split}
	\end{equation}
	\hrule
\end{figure*}
where $C_{\delta}^{(1)}(u,v)=\partial C_{\delta}(u,v)/\partial u$. Besides, $\Gamma(\cdot)$ and $\gamma(\cdot,\cdot)$ denote the Gamma function and the lower incomplete Gamma function, respectively.
\end{lemma}

\begin{proof}
	Since a symbol error occurs when $X^\star>\varpi_{0}[a_0]$, given $\varpi_{0}[a_0]=y$, we have $P_{\mathrm{e}}=\int_{-\infty}^{\infty}\left[1-F_{X^{\star}\mid\varpi_{0}[a_{0}]=y}(y)\right]$ $f_{\varpi_{0}[a_{0}]}(y)\mathrm{d}y$, which completes the proof.
\end{proof}

\begin{remark}
The dependence parameter $\delta$ in Lemma~\ref{l5} is determined from Kendall's tau \cite{bibitem3-1} between $\varpi_{0}[a_{0}]$ and $X^{\star}$ as 
$\delta={1}/{(1-\tau_{\mathrm{K}})}$,
where $\tau_{\mathrm{K}}$ is given by
\begin{equation}\label{kendall-tau}
		\tau_{\mathrm{K}}
		=
		\mathrm{E}\left[
		\operatorname{sgn}\left(
		\left(\varpi_{0}[a_{0}]-\widetilde{\varpi}_{0}[a_{0}]\right)
		\left(X^{\star}-\widetilde{X}^{\star}\right)
		\right)
		\right],
\end{equation}
which can be numerically evaluated using Algorithm~\ref{alg:delta}. The rationale and details are provided in Appendix~\ref{app:e}.
\end{remark}

\section{Simulation Results}
\begin{figure*}[t]
	\centering
	\subfloat[ $N=2^6$]{\includegraphics[width=0.32\textwidth]{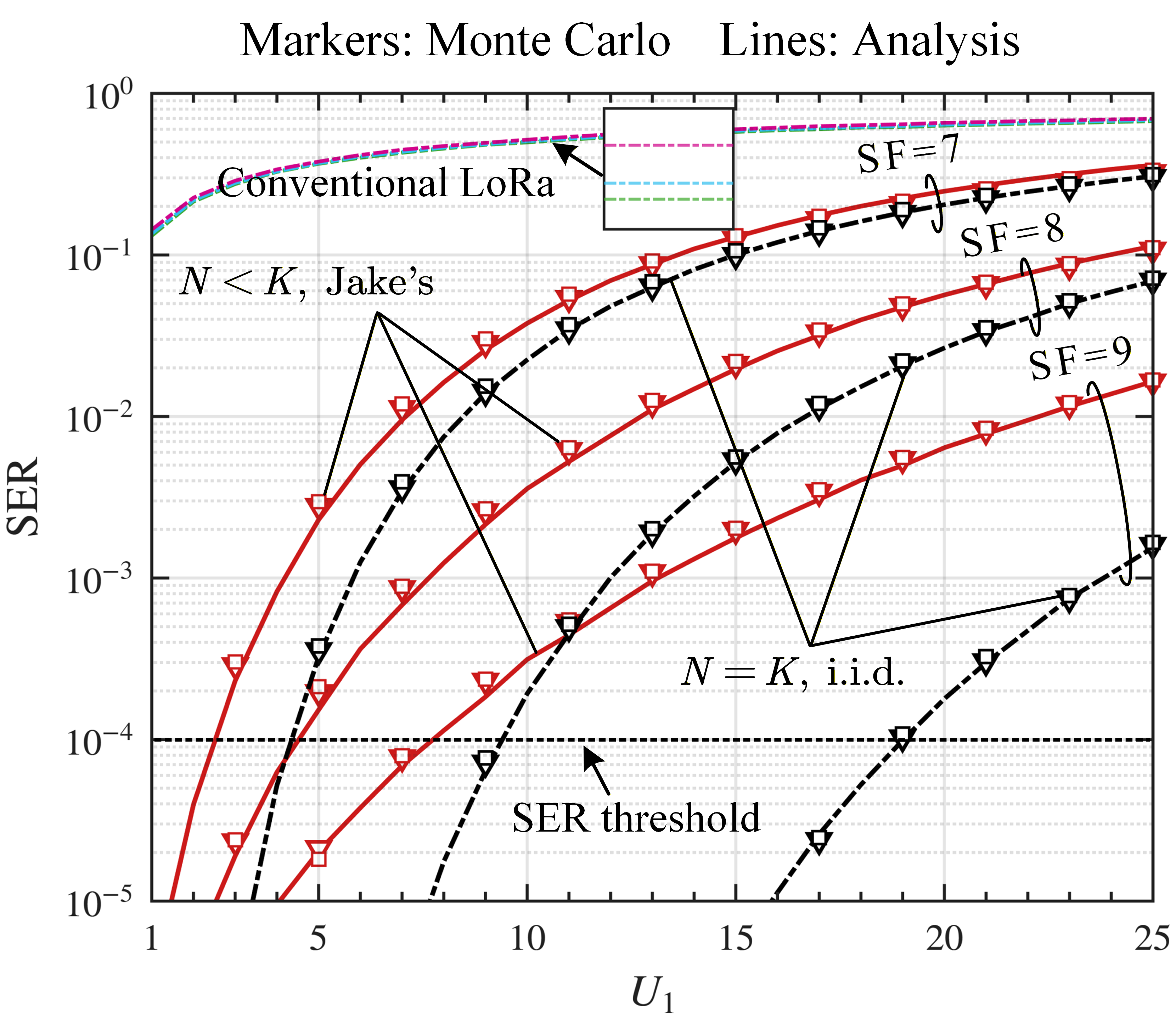}\label{Sim:1-1}}
	\hfill
	\subfloat[$N=2^7$]{\includegraphics[width=0.32\textwidth]{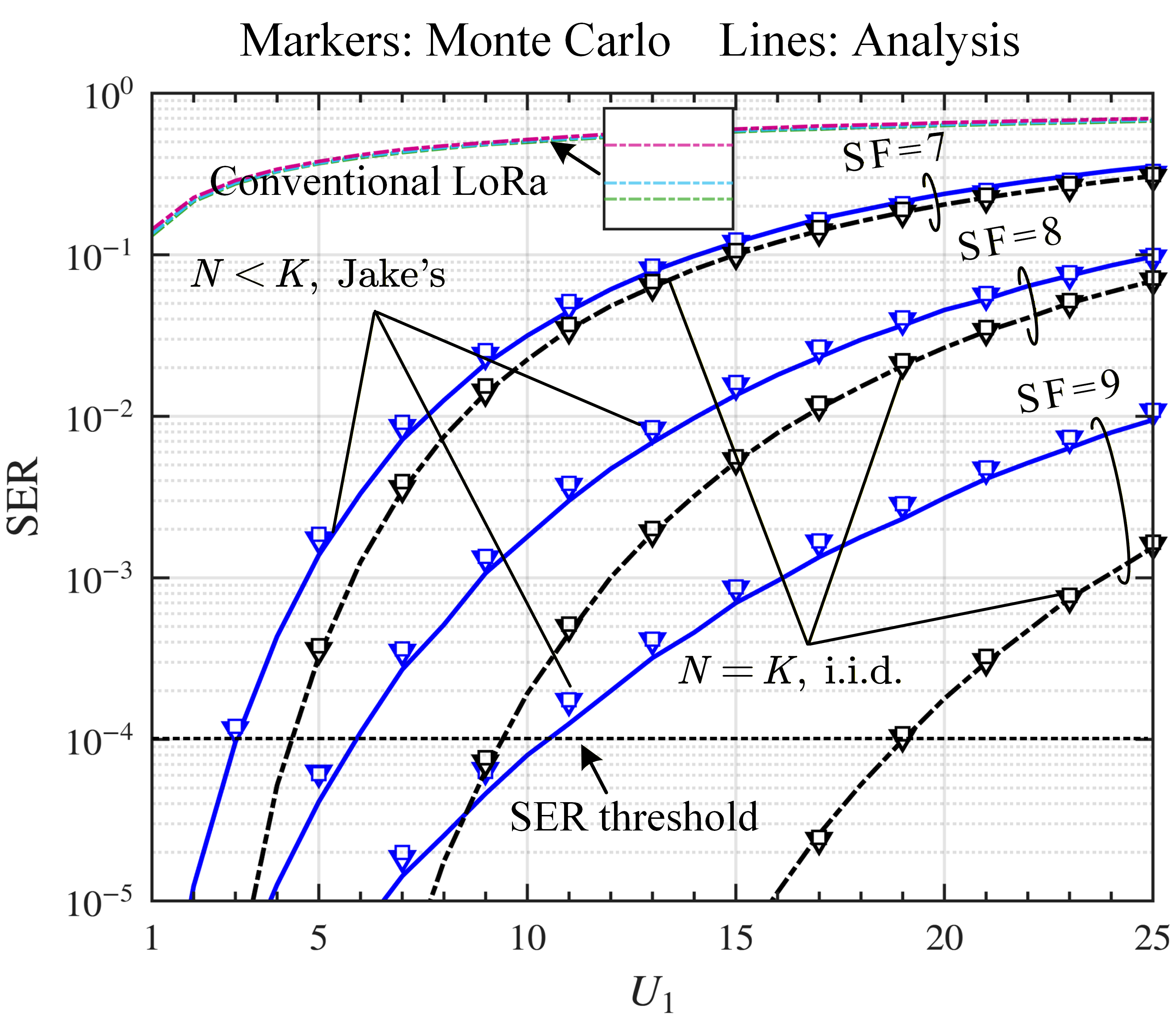}\label{Sim:1-2}}
	\hfill
	\subfloat[$N=2^8$]{\includegraphics[width=0.32\textwidth]{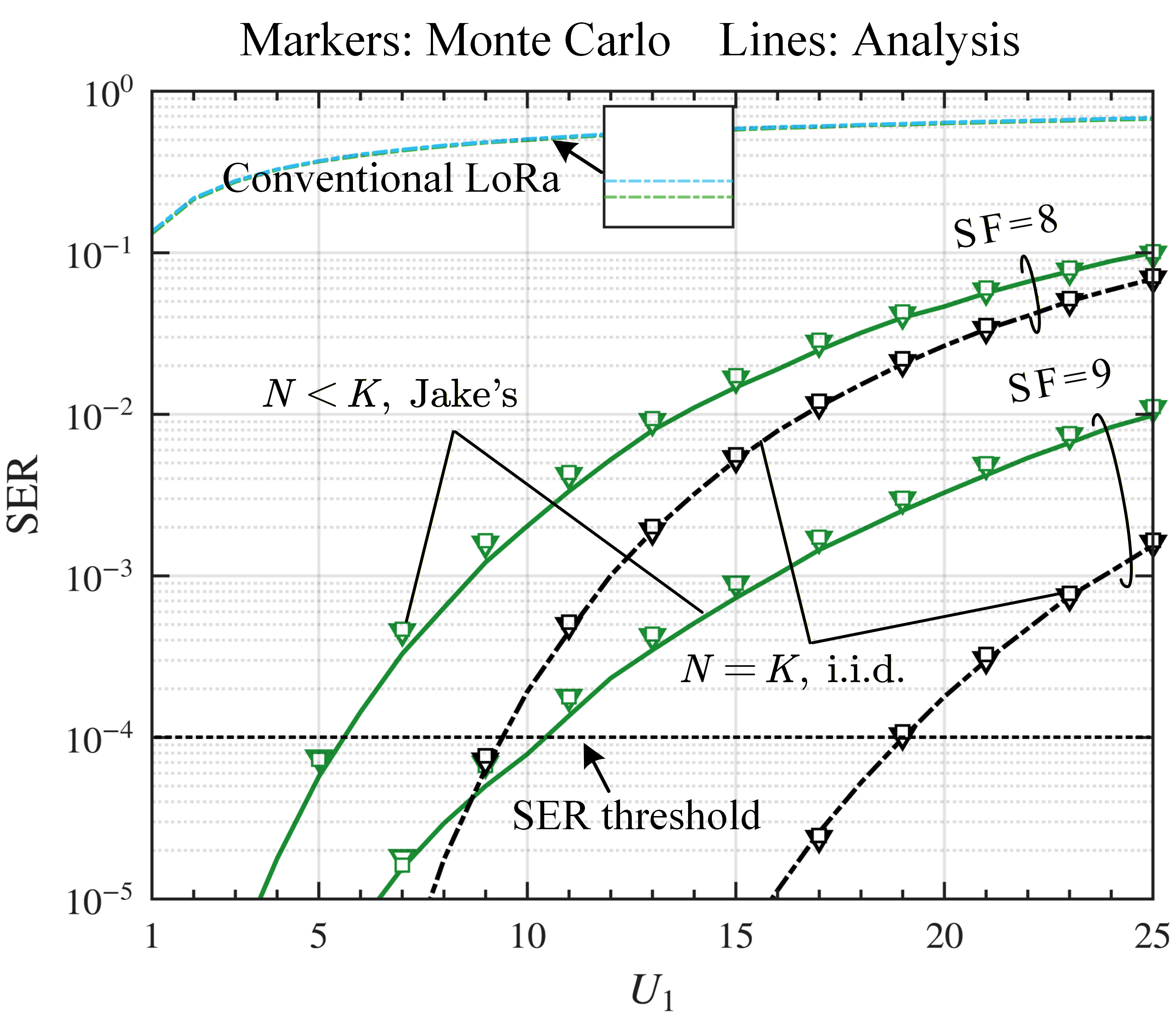}\label{Sim:1-3}}
	\caption{SER performance with $U_{1}$ interfering EDs. Square markers denote the Monte Carlo results obtained by neglecting the desired-signal contribution to the undesired bins, whereas inverted-triangle markers denote the exact system-level SER. Note that Fig.~\ref{Sim:1-3} only includes two set of results since $N<K$. ($SF=\{7,8,9\}$, $N=\{2^6,2^7,2^8\}$, $W_1=W_2=4$, $\mathrm{SNR}=0$~dB, $E_\mathrm{s}=1$, $\Omega=1$)}
	\label{Sim:1}
\end{figure*}

\begin{figure*}[t]
	\centering
	\subfloat[ $N=2^6$]{\includegraphics[width=0.32\textwidth]{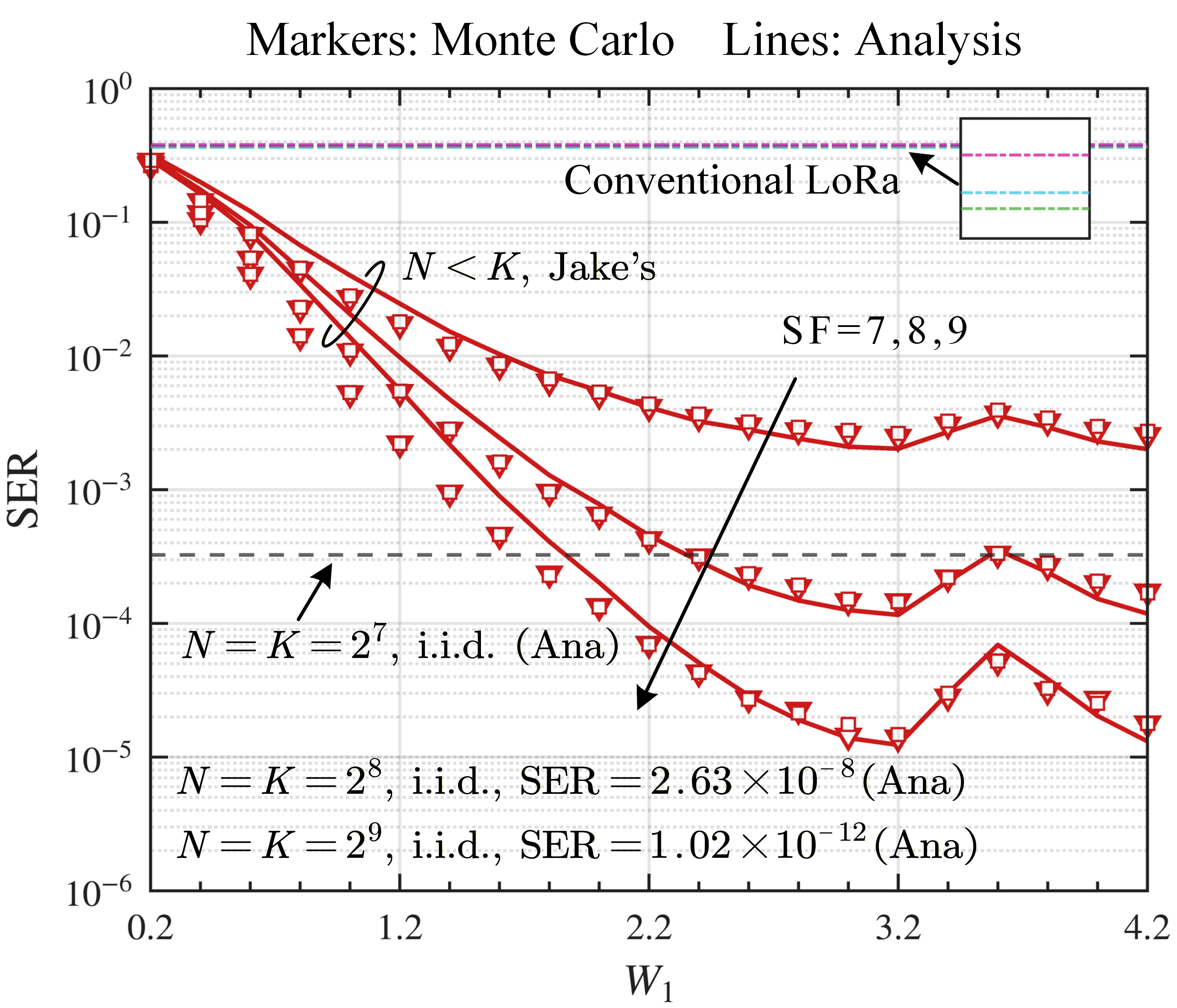}\label{Sim:2-1}}
	\hfill
	\subfloat[$N=2^7$]{\includegraphics[width=0.32\textwidth]{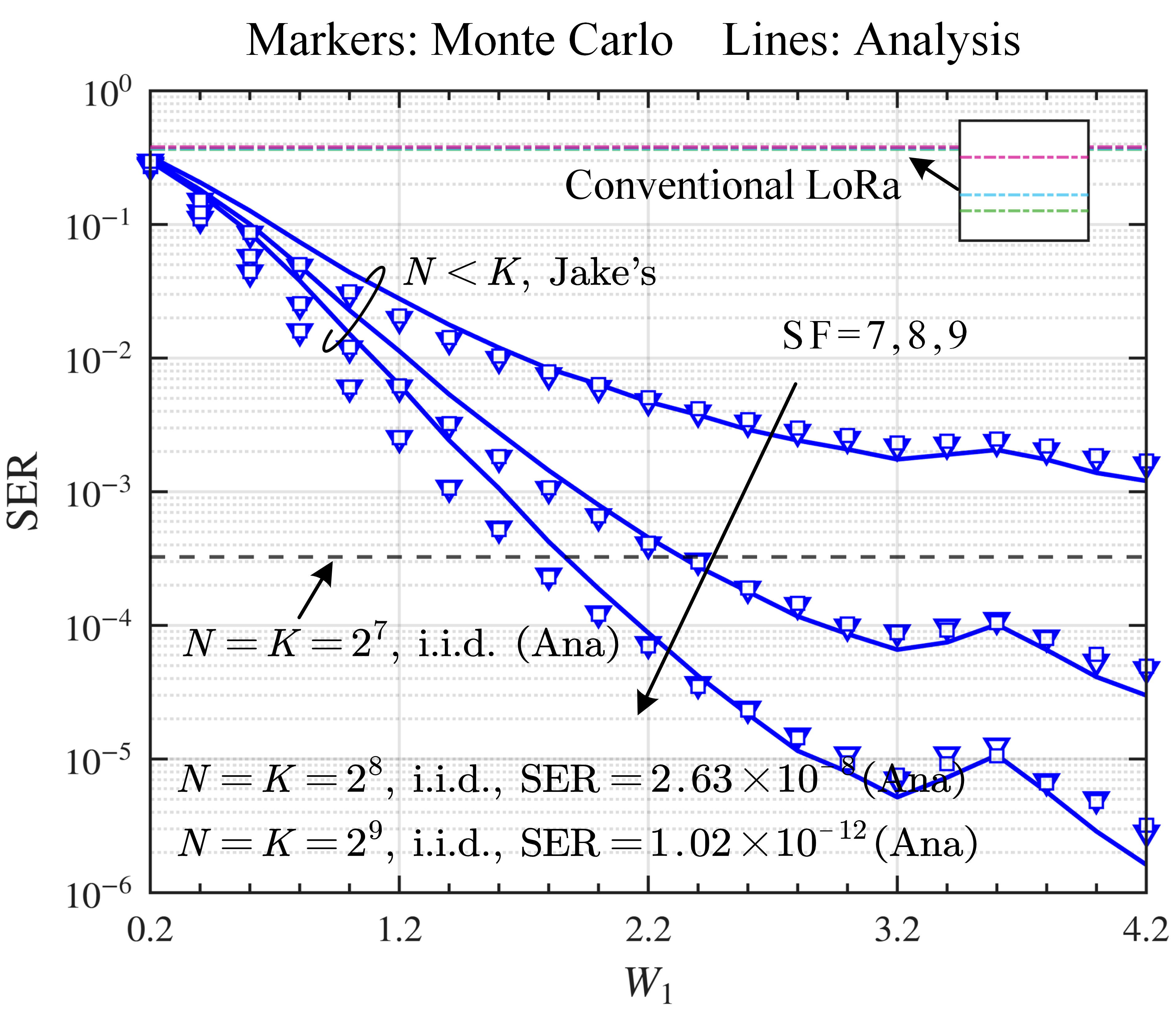}\label{Sim:2-2}}
	\hfill
	\subfloat[$N=2^8$]{\includegraphics[width=0.32\textwidth]{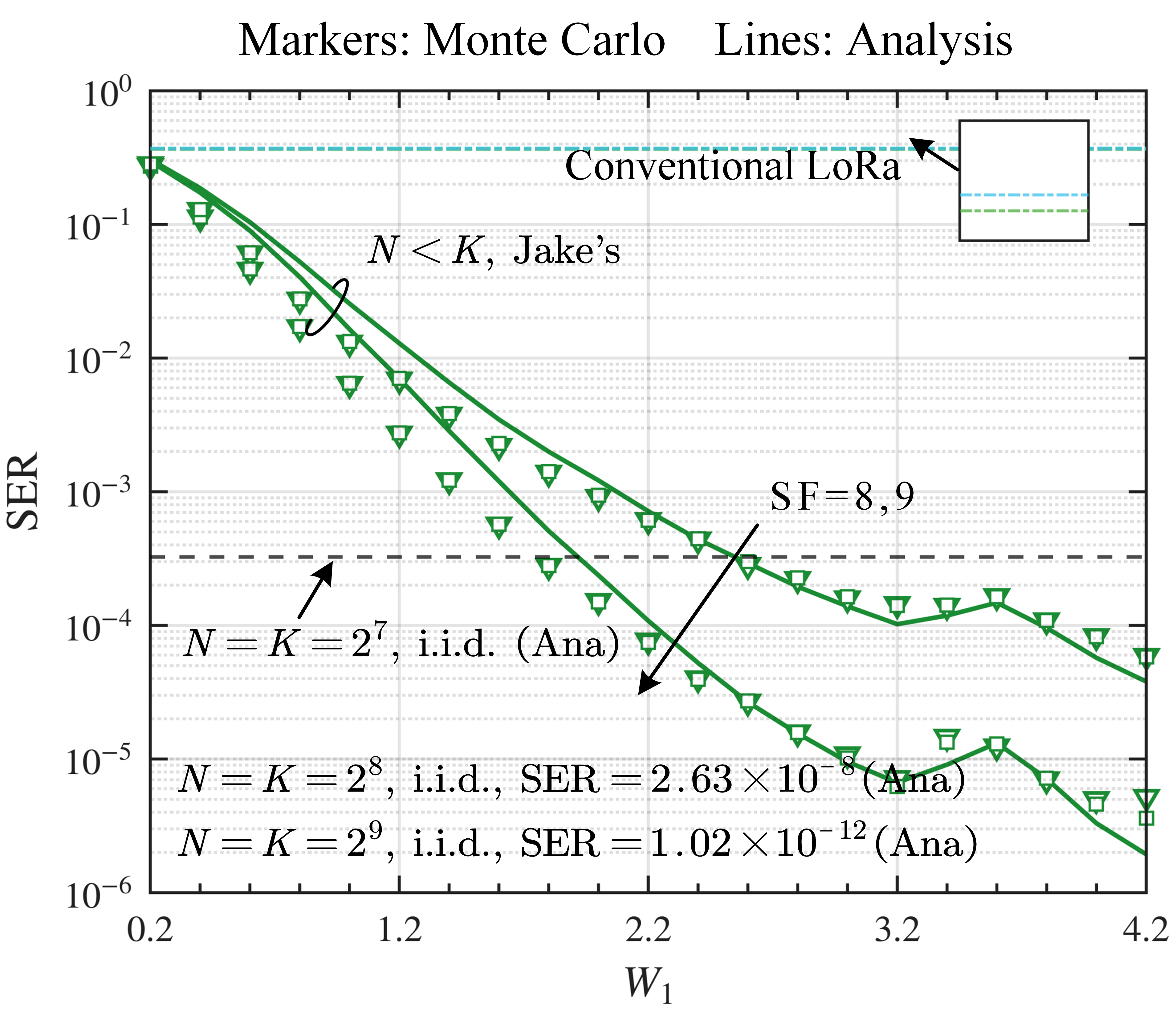}\label{Sim:2-3}}
	\caption{SER performance for different values of $W_{1}$ with fixed $W_{2}$. ($SF=\{7,8,9\}$, $N=\{2^6,2^7,2^8\}$, $U_1=5$, $W_2=4$, $\mathrm{SNR}=0$~dB, $E_\mathrm{s}=1$, $\Omega=1$)}
	\label{Sim:2}
\end{figure*}

\begin{figure*}[t]
	\centering
	\subfloat[ $N=2^6$]{\includegraphics[width=0.32\textwidth]{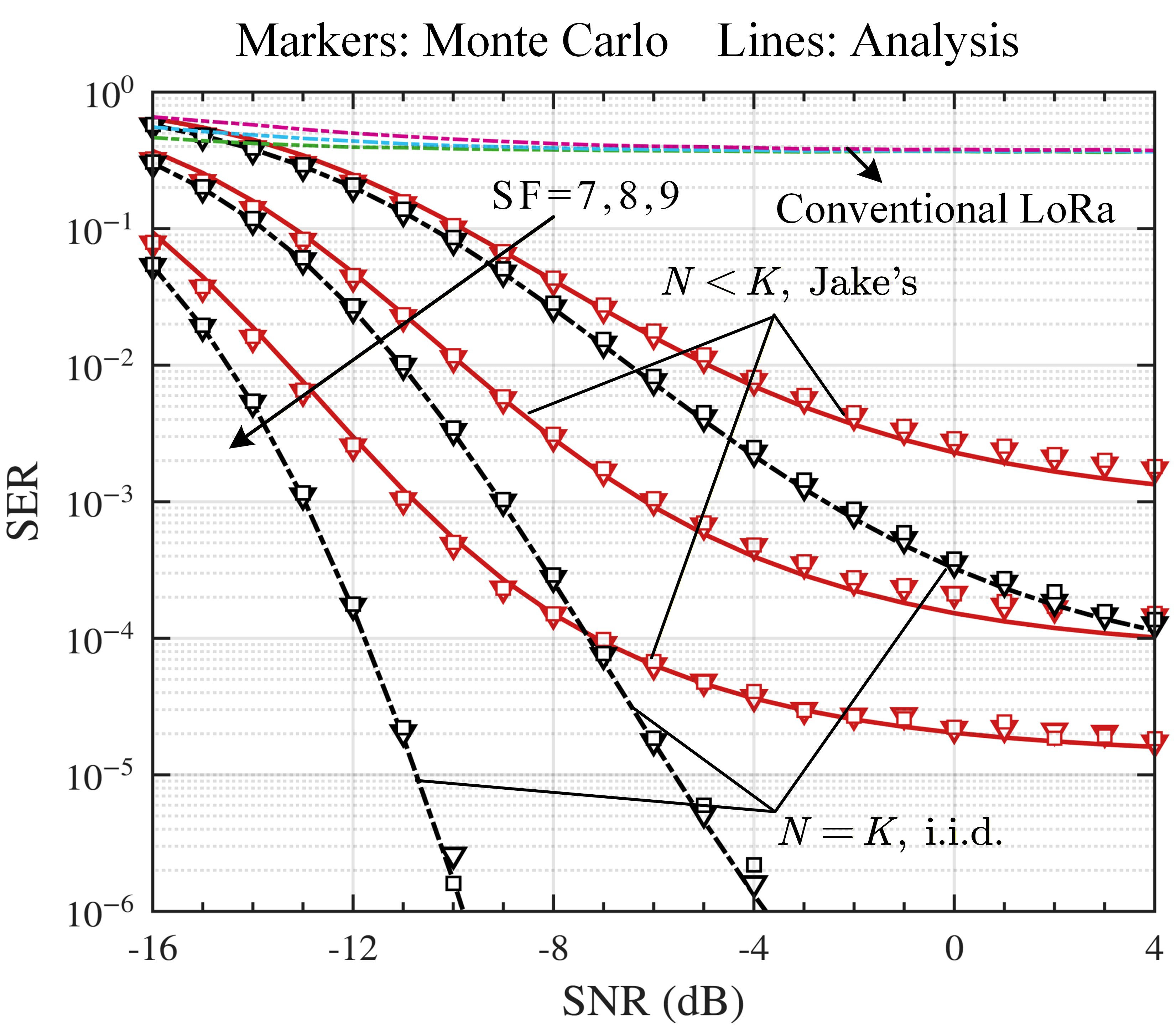}\label{Sim:3-1}}
	\hfill
	\subfloat[$N=2^7$]{\includegraphics[width=0.32\textwidth]{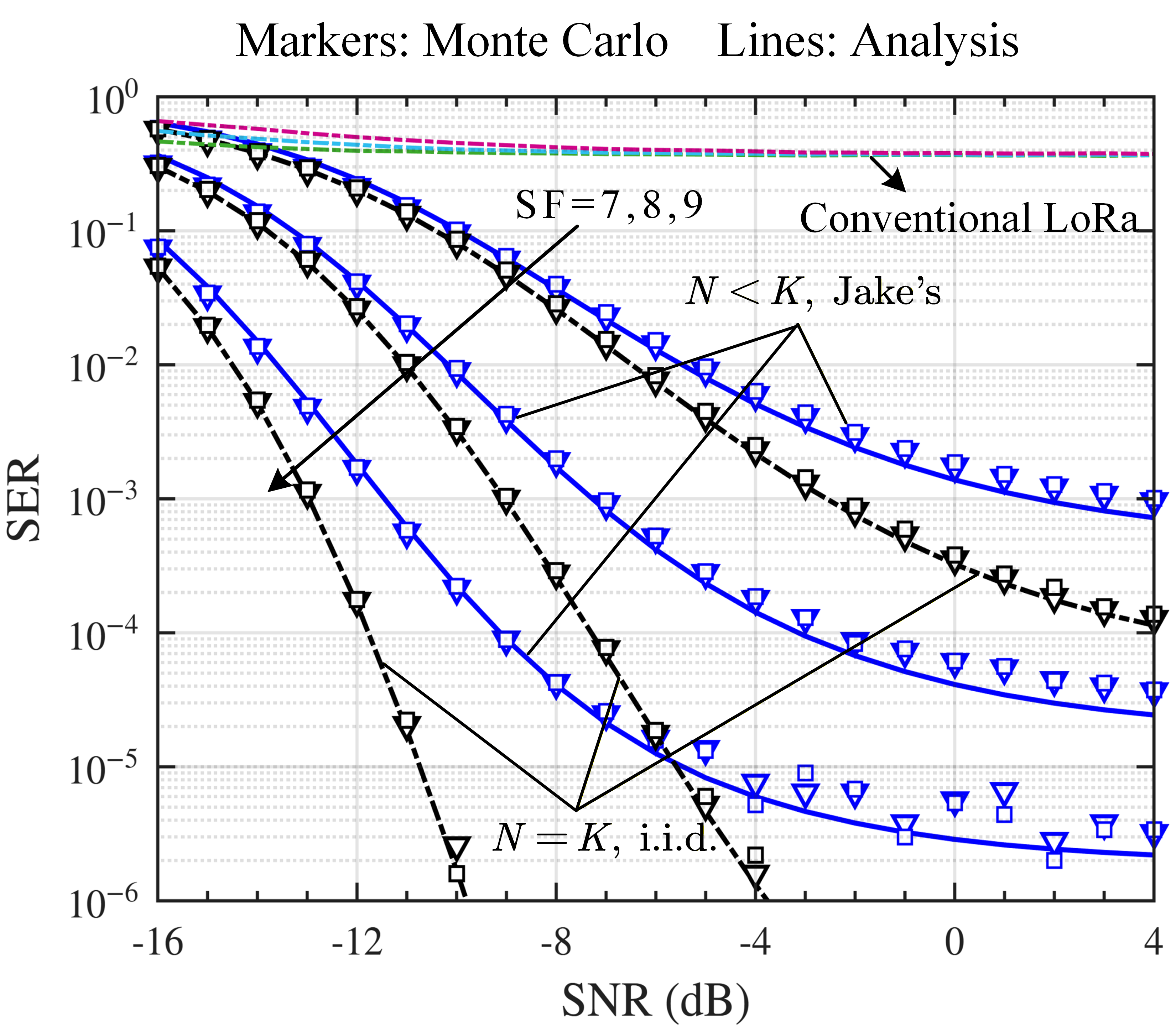}\label{Sim:3-2}}
	\hfill
	\subfloat[$N=2^8$]{\includegraphics[width=0.32\textwidth]{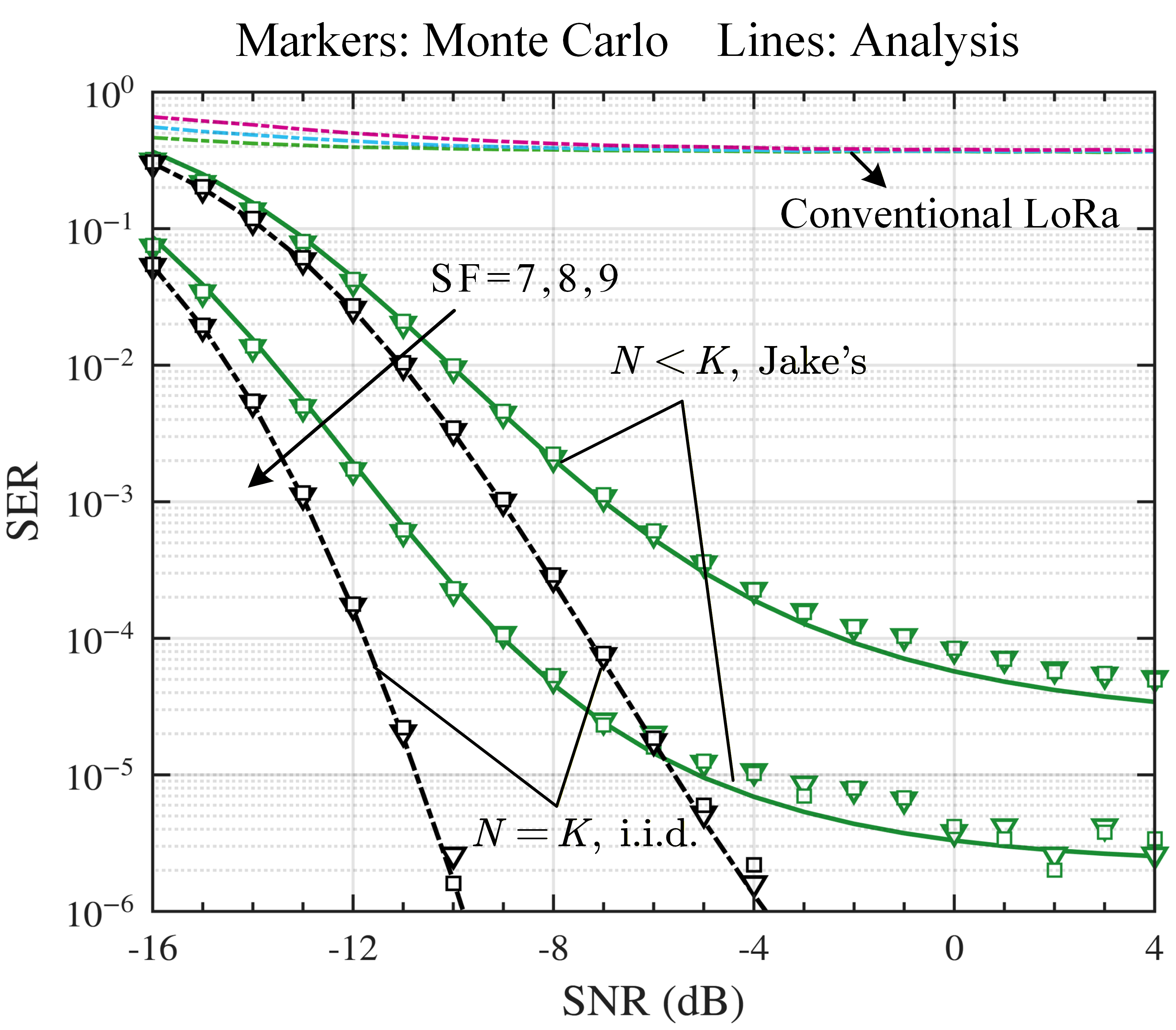}\label{Sim:3-3}}
	\caption{SER performance for different values of SNR. ($SF=\{7,8,9\}$, $N=\{2^6,2^7,2^8\}$, $U_1=5$, $W_1=W_2=4$, $E_\mathrm{s}=1$, $\Omega=1$)}
	\label{Sim:3}
\end{figure*}

In this section, we evaluate the performance of the proposed network via Monte Carlo simulations and analytical values. All parameters are specified in the figure captions. In the simulation, all users are placed at the same distance from the gateway to compare multiple-access performance under the same conditions.\footnote{This is consistent with typical LoRa network\cite{bibitem4-1,bibitem4-2}, where different SFs serve different concentric coverage rings. EDs within the same ring have similar distances to the gateway and thus comparable large-scale fading.} In addition, large-scale fading is not modeled separately since path loss can be incorporated into the SNR calculation \cite{bibitem4-3}. 
For example, with $P_\mathrm{t}=16$ dBm, $\Gamma=-16$ dB corresponds to a transmission distance of approximately $15.2$ km according to $L(d)=10\eta\log_{10}\left(4\pi f_\text{c}d/c\right)=P_\text{t}-\Gamma+174-10\log_{10}(\Delta)-NF$, where $\Delta=125$ kHz, $f_\text{c}=848$ MHz, $NF=6$ dB, $\eta=2.6$, and $c=3\times10^8$ m/s.
The previous and current interfering symbols are selected from four uniformly spaced values in $\mathcal{K}$ to maintain stable Monte Carlo results. Besides, each result is obtained from $5\times10^6$ Monte Carlo realizations.

In Fig.~\ref{Sim:1}, we study the relationship between SER and ED number. Conventional LoRa is used as a benchmark. Because it cannot distinguish the desired and interfering signals, the SER will exceed $10^{-1}$ for all ED values. Note that, as the interference is truncated within the observation window, two correlation peaks corresponding to the former and later symbols in \eqref{xj_piece} will separate the power. Thus, when $U_1=1$, the SER will be lower than $0.5$. Besides, a threshold of $\mathrm{SER}=10^{-4}$ is marked. The Monte Carlo results are also consistent with the simulation results. Based on the results, we have the following observation. First, compared with conventional LoRa, the {\it lora}-FAMA scheme can efficiently reduce interference among EDs. The SER can be reduced to below $10^{-4}$, thereby enabling simultaneous communication with multiple EDs, i.e., multi-user multiplexing. Second, three figures show the results for different $N$, where a smaller $N$ means the antenna can remain at one port for longer. In Fig.~\ref{Sim:1-3}, under $N=2^8$, a {\it lora}-FAMA with $\mathrm{SF}=9$ can serving more than $10$ EDs, each with a SER of $10^{-4}$. In Fig.~\ref{Sim:1-1}, a smaller $N=2^6$ can also serve more than $7$ EDs. Third, if we drop the physical size limitation, the ED number can be further increased. Larger $SF$ can support a greater number of EDs because the interference power will spread across more feasible frequency bins. Forth, for {\it lora}-FAMA, the enhancement from Fig.~\ref{Sim:1-2} to Fig.~\ref{Sim:1-3} is slower than that from Fig.~\ref{Sim:1-1} to Fig.~\ref{Sim:1-2}. Note that increasing $N$ increases the independence between undesired bins, which prevents the interfering ED from creating correlation peaks and thereby decreasing the SER. Then, the above results, with $N$ increasing from $2^7$ to $2^8$, will provide limited independence.

In Fig.~\ref{Sim:2}, we compare the SER performance under different $W_1$. First, as the length $W_1$ increases, the SER appears to decrease. In Fig.~\ref{Sim:2-1}, the {\it lora}-FAMA with $\mathrm{SF}=9$ can lower than $\mathrm{SER}=1^{-4}$, when $W_1>1.6$. Second, the SER does not consistently decrease as $W_1$ increases. For example, from $W_1=3.2$ to $3.6$, the SER slightly increased. This is because the overall correlation among undesired bins is not always reduced in Jake's model. Similar to the behavior of the Bessel function in \eqref{jake-model}, the sum of the absolute correlation coefficients generally decreases with $W_{1}$ while exhibiting local oscillations, with a correlation peak at $W_{1}=3.6$. Third, bounds for different SFs are given, and the SER performance will not fall below these bounds. For example, in Fig.~\ref{Sim:2-1}, the SER of the case $\mathrm{SF}=7$ will not be lower than $3.2\time10^{-4}$ at this configuration, while the SER of the case $\mathrm{SF}=7$ in Fig.~\ref{Sim:2-2} will approach this limitation as $W_1$ increases. The other two bounds for $\mathrm{SF}=\{8,9\}$ are beyond the range of the Monte Carlo method. Therefore, we only provide the analytical values.

In Fig.~\ref{Sim:3}, the SER performance under different SNRs is evaluated. First, as the SNR increases, the SER decreases and gradually approaches an error floor determined by the interfering EDs. Second, note that the length of the symbol $K$ mainly influences the signal power. Therefore, increasing $N$ shows no obvious decrease in the low-SNR region. Third, it is observed that, in the low SNR region, e.g., at $-12$~dB, the case that $N=K$, $\mathbf{R}=\mathbf{I}_K$ also shows significant decreasing in SNR, this is because independent channel components provide a stronger hardening effect in the desired bin, which reduces the variation of the combined channel power and improves the reliability of signal detection.

\section{Conclusion}
This paper proposed a novel FAMA framework to enable simultaneous multiuser transmissions in LoRa networks, referred to as {\it lora}-FAMA. Inspired by the fact that the received LoRa symbol over a variable channel can not be demodulated. We continuously allow the fluid antenna to traverse the aperture during communication. Then the proposed receiver coherently accumulates the desired signal while preventing the interfering signals from being demodulated. Both analytical and Monte Carlo results demonstrate the effectiveness of the proposed scheme. Under a target SER of $10^{-4}$, a fluid antenna with a normalized area of  $4\times4$ can serve more than $5$ EDs at $SF=8$ and more than $10$ EDs at $SF=9$. Without constraining the physical size of the gateway antenna, the number of concurrently supported users increases to 9 and 19, respectively.
Furthermore, to reduce the switching overhead, the FAS may remain at each port for multiple consecutive samples rather than switching at every sampling instant. For SF $=9$, even when the FAS dwells at each port for four consecutive samples, {\it lora}-FAMA can still support $7$ concurrent users. Despite some practical issues requiring further investigation, the proposed scheme shows significant potential in massive IoT connectivity.

\appendices
\section{Proof of {\it Lemma~\ref{l1}}}\label{app:a}
At $k=a_{0}$, the desired component can be rewritten as $\varpi_{\mathrm{s},0}[a_{0}]=\Re\left\{w_{\mathrm{s},0}[a_{0}]\right\}=\sqrt{E_{\mathrm{s}}}\mathbf{h}_{0}^{\mathrm{H}}\mathbf{h}_{0}/N$. Since $\mathbf{h}_{0}\sim\mathcal{CN}(\mathbf{0},\Omega\mathbf{R})$ and $\mathrm{tr}(\mathbf{R})=N$, while the interference and noise components are zero-mean, we obtain
\begin{equation}
	\begin{split}
		\mathrm{E}\left[\varpi_{0}[a_{0}]\right]&=\frac{\sqrt{E_{\mathrm{s}}}}{N}\mathrm{E}\left[\mathbf{h}_{0}^{\mathrm{H}}\mathbf{h}_{0}\right]=\sqrt{E_{\mathrm{s}}}\Omega.
	\end{split}
\end{equation}
Moreover, the variance of the desired-signal component is
\begin{equation}\label{var-s}
	\begin{split}
		\mathrm{Var}\big[\varpi_{\mathrm{s},0}[a_{0}]\big]=\frac{E_{\mathrm{s}}}{N^{2}}\mathrm{Var}\left[\mathbf{h}_{0}^{\mathrm{H}}\mathbf{h}_{0}\right]=\frac{E_{\mathrm{s}}\Omega^{2}}{N^{2}}\mathrm{tr}\left(\mathbf{R}^{2}\right).
	\end{split}
\end{equation}

For the $j$-th interfering ED, interference term at the desired bin is given by $w_{\mathrm{I},0,j}[a_{0}]=\sqrt{E_{\mathrm{s}}}\mathbf{h}_{0}^{\mathrm{H}}\mathbf{D}_{a_{0}\mid\tau,\boldsymbol{\upsilon}_{j}}\mathbf{h}_{j}$. Since $\mathbf{h}_{0}$ and $\mathbf{h}_{j}$ are independent and follow $\mathcal{CN}(\mathbf{0},\Omega\mathbf{R})$ and the interference term is zero-mean proper complex, we have 
\begin{equation}
	\begin{split}
		&\mathrm{Var}\left[\varpi_{\mathrm{I},0,j}[a_{0}]\right]=\frac{1}{2}\mathrm{E}_{\tau,\boldsymbol{\upsilon}_{j}}\left[\mathrm{E}\left[\left|w_{\mathrm{I},0,j}[a_{0}]\right|^{2}\mid\tau,\boldsymbol{\upsilon}_{j}\right]\right]\\
		&\quad=\frac{E_{\mathrm{s}}\Omega^{2}}{2QK^{3}}\sum_{\tau=0}^{Q_{1}}\sum_{\boldsymbol{\upsilon}\in\mathcal{K}^{3}}\mathrm{tr}\left(\mathbf{R}\mathbf{D}_{a_{0}\mid\tau,\boldsymbol{\upsilon}}\mathbf{R}\mathbf{D}_{a_{0}\mid\tau,\boldsymbol{\upsilon}}^{\mathrm{H}}\right)=\frac{E_{\mathrm{s}}\Omega^{2}}{2K}.
	\end{split}
\end{equation}
The last equality follows from
\begin{equation}
	\begin{split}
		&\frac{1}{K^3}\sum_{\boldsymbol{\upsilon}\in\mathcal{K}^3}\mathrm{tr}\left(\mathbf{R}\mathbf{D}_{a_0\mid\tau,\boldsymbol{\upsilon}}\mathbf{R}\mathbf{D}_{a_0\mid\tau,\boldsymbol{\upsilon}}^{\mathrm{H}}\right)\\
		&=\sum_{l=1}^{N}\sum_{m=1}^{N}
		\left|[\mathbf R]_{l,m}\right|^2
		\underbrace{\frac{1}{K^3}\sum_{\boldsymbol{\upsilon}\in\mathcal{K}^3}\left[\mathbf{D}_{k\mid\tau,\boldsymbol{\upsilon}}\right]_{m,m}\left[\mathbf{D}_{k\mid\tau,\boldsymbol{\upsilon}}\right]_{l,l}^{*}}_{\frac{1}{NK}\mathbf{1}_{\{l=m\}}}=\frac{1}{K},
	\end{split}
\end{equation}
and $[\mathbf{R}]_{l,l}=1$. Since the $U_{1}$ interfering EDs have identical variances and are mutually uncorrelated, the total interference variance is
\begin{equation}\label{var-I}
	\begin{split}
		\mathrm{Var}\left[\varpi_{\mathrm{I},0}[a_{0}]\right]=\frac{E_{\mathrm{s}}\Omega^{2}U_{1}}{2K}.
	\end{split}
\end{equation}

For the noise component, we have $w_{\mathrm{n},0}[a_{0}]=\mathbf{h}_{0}^{\mathrm{H}}\boldsymbol{\eta}$,
where $\boldsymbol{\eta}\sim\mathcal{CN}\left(\mathbf{0},{N_{0}}\mathbf{I}_{N}/{N}\right)$. It follows that
\begin{equation}\label{var-n}
\begin{split}
	\mathrm{Var}\left[\varpi_{\mathrm{n},0}[a_{0}]\right]=\frac{1}{2}\mathrm{E}\left[\left|w_{\mathrm{n},0}[a_{0}]\right|^{2}\right]=\frac{\Omega N_{0}}{2}.
\end{split}
\end{equation}
Moreover, we have
\begin{equation}\label{cross-cov}
	\begin{split}
		&\mathrm{Cov}\left[\varpi_{\mathrm{s},0}[a_{0}],\varpi_{\mathrm{I},0,j}[a_{0}]\right]\\
		&\:=\mathrm{E}_{\mathbf{h}_{0},\tau,\boldsymbol{\upsilon}_{j}}\left[\left(\varpi_{\mathrm{s},0}[a_{0}]-\mathrm{E}\left[\varpi_{\mathrm{s},0}[a_{0}]\right]\right)\right.\\
		&\:\qquad\qquad\left.\times\mathrm{E}_{\mathbf{h}_{j}}\left[\varpi_{\mathrm{I},0,j}[a_{0}]\mid\mathbf{h}_{0},\tau,\boldsymbol{\upsilon}_{j}\right]\right]\\
		&\:=\mathrm{E}_{\mathbf{h}_{0},\tau,\boldsymbol{\upsilon}_{j}}\left[\left(\varpi_{\mathrm{s},0}[a_{0}]-\mathrm{E}\left[\varpi_{\mathrm{s},0}[a_{0}]\right]\right)\right.\\
		&\:\qquad\qquad\left.\times\Re\left\{\sqrt{E_{\mathrm{s}}}\mathbf{h}_{0}^{\mathrm{H}}\mathbf{D}_{a_{0}\mid\tau,\boldsymbol{\upsilon}_{j}}\mathrm{E}_{\mathbf{h}_{j}}\left[\mathbf{h}_{j}\right]\right\}\right]=0.
	\end{split}
\end{equation}
where the last equality follows from $\mathrm{E}_{\mathbf{h}_{j}}\left[\mathbf{h}_{j}\right]=\mathbf{0}$. Then,
$\mathrm{Cov}\left(\varpi_{\mathrm{s},0}\big[a_{0}\big],\varpi_{\mathrm{n},0}\big[a_{0}\big]\right)$, $\mathrm{Cov}\left(\varpi_{\mathrm{I},0,j}\big[a_{0}\big],\varpi_{\mathrm{n},0}\big[a_{0}\big]\right)$, $\mathrm{Cov}\left(\varpi_{\mathrm{I},0,j_1}\big[a_{0}\big],\varpi_{\mathrm{I},0,j_2}\big[a_{0}\big]\right)$, where $j_1\neq{j}_2$, all equal to zero.

Finally, combining \eqref{var-s}, \eqref{var-I}, and \eqref{var-n} yields \eqref{var_wi}. This completes the proof.\qed

\section{Proof of {\it Lemma~\ref{l2}}}\label{app:b}

For $N=K$ and large $W$, we have $\mathbf{R}\to\mathbf{I}_{K}$ and each FAS port observes one sample. Moreover, $\mathbf{D}_{a_{0}\mid\tau,\boldsymbol{\upsilon}_{j}}\mathbf{D}_{a_{0}\mid\tau,\boldsymbol{\upsilon}_{j}}^{\mathrm{H}}=\mathbf{I}_{K}/K^{2}$, such that the distribution of $\varpi_{0}[a_{0}]$ is independent of $\tau$ and $\boldsymbol{\upsilon}_{j}$. Let $\zeta_{0}=\mathbf{h}_{0}^{\mathrm{H}}\mathbf{h}_{0}\sim\mathrm{Gamma}(K,\Omega)$. The desired-signal component is then given by $\varpi_{\mathrm{s},0}[a_{0}]=\sqrt{E_{\mathrm{s}}}\zeta_{0}/K=\theta\zeta_{0}$.

Conditioned on $\zeta_{0}$, the interference and noise terms are zero-mean Gaussian with variance $v\zeta_{0}$, where $v=E_{\mathrm{s}}\Omega U_{1}/(2K^{2})+N_{0}/(2K)$. Hence,
\begin{equation}
	\begin{split}
		\varpi_{0}[a_{0}]\mid\zeta_{0}\sim\mathcal{N}\left(\theta \zeta_{0},v\zeta_{0}\right),
	\end{split}
\end{equation}
where $\theta=\sqrt{E_{\mathrm{s}}}/K$. The CF of $\varpi_{0}[a_{0}]$ is therefore
\begin{equation}\label{w-cf}
	\begin{split}
		\Psi_{\varpi_{0}[a_{0}]}(t)&=\mathrm{E}_{\zeta_{0}}\left[\exp\left(\left(jt\theta-\frac{t^{2}v}{2}\right)\zeta_{0}\right)\right]\\
		&=\left[1-\Omega\left(jt\theta-\frac{t^{2}v}{2}\right)\right]^{-K}\\
		&=\left(1-jpt\right)^{-K}\left(1+jqt\right)^{-K}.
	\end{split}
\end{equation}
The last equality follows from $p-q=\theta\Omega$ and $pq=\Omega v/2$. Therefore, $\varpi_{0}[a_{0}]$ has the same distribution as $Y_{1}-Y_{2}$, where $Y_{1}\sim\mathrm{Gamma}(K,p)$ and $Y_{2}\sim\mathrm{Gamma}(K,q)$ are independent\cite[Chapter 22]{app:b-1}.

For $x\geq0$, the CDF can be written as
\begin{equation}\label{int_x>0}
	\begin{split}
		F_{\varpi_{0}[a_{0}]}(x)&=\Pr(Y_{1}-Y_{2}\leq x)=\int_{0}^{\infty}F_{Y_{1}}(x+y)f_{Y_{2}}(y)\mathrm{d}y\\
		&=\frac{1}{\Gamma(K)q^{K}}\int_{0}^{\infty}F_{Y_{1}}(x+y)y^{K-1}\exp\left(-\frac{y}{q}\right)\mathrm{d}y.
	\end{split}
\end{equation}
Note that $K$ is a integer, then we have $F_{Y_{1}}(x+y)=1-\exp\left(-(x+y)/p\right)\sum_{\kappa=0}^{K_1}(x+y)^{\kappa}/(\kappa!p^{\kappa})$. Substituting it into \eqref{int_x>0} and expanding $(x+y)^{\kappa}$ yield
\begin{equation}
	\begin{split}
		F_{\varpi_{0}[a_{0}]}(x)=&1-\frac{\exp\left(-x/p\right)}{\Gamma(K)q^{K}}\sum_{\kappa=0}^{K_1}\sum_{b=0}^{\kappa}\binom{\kappa}{b}\frac{x^{\kappa-b}}{\kappa!p^{\kappa}}\\
		&\times\int_{0}^{\infty}y^{K+b-1}\exp(-\rho y)\mathrm{d}y\\
		=&1-\exp\left(-\frac{x}{p}\right)\sum_{\kappa=0}^{K_1}\sum_{b=0}^{\kappa}A_{\kappa,b}\frac{x^{\kappa-b}}{p^{\kappa}q^{K}},
	\end{split}
\end{equation}
where $\int_{0}^{\infty}y^{K+b-1}\exp(-\rho y)\mathrm{d}y=\Gamma(K+b)/\rho^{K+b}$. Hence, \eqref{a_i-x>0} is obtained.

For $x<0$, let $u=-x>0$. The corresponding CDF satisfies $F_{\varpi_{0}[a_{0}]}(-u)=\Pr(Y_{1}+u\leq Y_{2})$, which gives
\begin{equation}
	\begin{split}
		F_{\varpi_{0}[a_{0}]}(-u)&=\frac{\exp\left(-u/q\right)}{\Gamma(K)p^{K}}\\
		&\times\sum_{\kappa=0}^{K_1}\frac{1}{\kappa!q^{\kappa}}\int_{0}^{\infty}(u+y)^{\kappa}y^{K-1}\exp(-\rho y)\mathrm{d}y\\
		&=\exp\left(-\frac{u}{q}\right)\sum_{\kappa=0}^{K_1}\sum_{b=0}^{\kappa}A_{\kappa,b}\frac{u^{\kappa-b}}{p^{K}q^{\kappa}}.
	\end{split}
\end{equation}
Substituting $u=-x$ yields \eqref{a_i-x<0}. This completes the proof.\qed

\section{Proof of {\it Lemma~\ref{l3}}}\label{app:c}

Given $\tau$, $\{\boldsymbol{\upsilon}_{j}\}_{j=1}^{U_1}$ and $\mathbf{h}_{0}$, the interference of the $j$-th ED is given by $w_{\mathrm{I},0,j}[k]=\sqrt{E_{\mathrm{s}}}\mathbf{h}_0^{\mathrm{H}}\mathbf{D}_{k\mid\tau,\boldsymbol{\upsilon}_{j}}\mathbf{h}_{j}$. Since $\mathbf{h}_{j}\sim\mathcal{CN}(\mathbf{0},\Omega\mathbf{R})$, its real part satisfies $\varpi_{\mathrm{I},0,j}[k]\mid\tau,\boldsymbol{\upsilon}_{j},\mathbf{h}_0\sim\mathcal{N}\big(0,\frac{E_{\mathrm{s}}\Omega}{2}\mathbf{h}_0^{\mathrm{H}}\mathbf{D}_{k\mid\tau,\boldsymbol{\upsilon}_{j}}\mathbf{R}\mathbf{D}_{k\mid\tau,\boldsymbol{\upsilon}_{j}}^{\mathrm{H}}\mathbf{h}_0\big)$.
Moreover, the noise term is given by $w_{\mathrm{n},0}[k]=\mathbf{h}_0^{\mathrm{H}}\boldsymbol{\eta}_{k}$, where $\boldsymbol{\eta}_{k}\sim\mathcal{CN}\left(\mathbf{0},N_{0}\mathbf{I}_{N}/N\right)$. Hence, $\varpi_{\mathrm{n},0}[k]\mid\mathbf{h}_0\sim\mathcal{N}\left(0,N_{0}\mathbf{h}_0^{\mathrm{H}}\mathbf{h}_0/(2N)\right)$.
Note that the interfering channels and noise are mutually independent, its conditional CF is\cite{app:c-1}
\begin{equation}\label{cf-u}
	\begin{split}
		&\Psi_{X_{k}\mid\tau,\{\boldsymbol{\upsilon}_{j}\}_{j=1}^{U_1},\mathbf{h}_0}(t)=\mathrm{E}\left[\exp\left(jtX_{k}\right)\mid\tau,\{\boldsymbol{\upsilon}_{j}\}_{j=1}^{U_1},\mathbf{h}_{0}\right]\\
		&=\exp\bigg[-\frac{t^{2}}{4}\mathbf{h}_0^{\mathrm{H}}\bigg(E_{\mathrm{s}}\Omega\sum_{j=1}^{U_{1}}\mathbf{D}_{k\mid\tau,\boldsymbol{\upsilon}_{j}}\mathbf{R}\mathbf{D}_{k\mid\tau,\boldsymbol{\upsilon}_{j}}^{\mathrm{H}}+\frac{N_{0}}{N}\mathbf{I}_{N}\bigg)\mathbf{h}_0\bigg].
	\end{split}
\end{equation}

Averaging \eqref{cf-u} over $\mathbf{h}_{0}\sim\mathcal{CN}(\mathbf{0},\Omega\mathbf{R})$ gives \eqref{cf-h} at the top of this page,
\begin{figure*}
\begin{equation}\label{cf-h}
	\begin{split}
		\Psi_{X_{k}\mid\tau,\{\boldsymbol{\upsilon}_{j}\}_{j=1}^{U_1}}(t)
		&=\frac{1}{\pi^{N}\Omega^{N}\det(\mathbf{R})}\int_{\mathbb{C}^{N}}\exp\Bigg[-\mathbf{u}^{\mathrm{H}}\bigg(\frac{1}{\Omega}\mathbf{R}^{-1}+\frac{t^{2}N_{0}}{4N}\mathbf{I}_{N}+\frac{t^{2}E_{\mathrm{s}}\Omega}{4}\sum_{j=1}^{U_{1}}\mathbf{D}_{k\mid\tau,\boldsymbol{\upsilon}_{j}}\mathbf{R}\mathbf{D}_{k\mid\tau,\boldsymbol{\upsilon}_{j}}^{\mathrm{H}}\bigg)\mathbf{u}\Bigg]\mathrm{d}^{2N}\mathbf{u}\\
		&=\det\bigg(\mathbf{M}_{\mathrm{n}}(t)+\frac{t^{2}}{4}\sum_{j=1}^{U_{1}}\mathbf{B}_{k\mid\tau,\boldsymbol{\upsilon}_{j}}\bigg)^{-1},
	\end{split}
\end{equation}
	\hrule
\end{figure*}
where $\mathbf{M}_{\mathrm{n}}(t)=\mathbf{I}_{N}+t^{2}N_{0}\Omega\mathbf{R}/(4N)$ and $\mathbf{B}_{k\mid\tau,\boldsymbol{\upsilon}_{j}}=E_{\mathrm{s}}\Omega^{2}\mathbf{R}^{\frac{1}{2}}\mathbf{D}_{k\mid\tau,\boldsymbol{\upsilon}_{j}}\mathbf{R}\mathbf{D}_{k\mid\tau,\boldsymbol{\upsilon}_{j}}^{\mathrm{H}}\mathbf{R}^{\frac{1}{2}}$. The last equality follows from the complex Gaussian integral and $\det(\mathbf{I}+\mathbf{A}\mathbf{C})=\det(\mathbf{I}+\mathbf{C}\mathbf{A})$.

Averaging \eqref{cf-h} over the discrete variables of the $U_{1}$ interfering EDs gives the exact CF conditioned on $\tau$ as
\begin{equation}\label{cf-e}
	\begin{split}
		\Psi_{X_{k}\mid\tau}(t)={}&\frac{1}{K^{3U_{1}}}\sum_{\boldsymbol{\upsilon}_{1}\in\mathcal{K}^{3}}\cdots\sum_{\boldsymbol{\upsilon}_{U_{1}}\in\mathcal{K}^{3}}\\
		&\det\bigg(\mathbf{M}_{\mathrm{n}}(t)+\frac{t^{2}}{4}\sum_{j=1}^{U_{1}}\mathbf{B}_{k\mid\tau,\boldsymbol{\upsilon}_{j}}\bigg)^{-1}.
	\end{split}
\end{equation}
Evaluating \eqref{cf-e} requires examining all $K^{3U_{1}}$ combinations. For notation compactness, define $\widetilde{\mathbf{B}}_{k\mid\tau,\boldsymbol{\upsilon}_{j}}(t)=\mathbf{M}_{\mathrm{n}}^{-\frac{1}{2}}(t)\mathbf{B}_{k\mid\tau,\boldsymbol{\upsilon}_{j}}\mathbf{M}_{\mathrm{n}}^{-\frac{1}{2}}(t)$. The determinant in \eqref{cf-e} is approximated as follows.

\begin{figure*}[!t]
	\normalsize
	\begin{subequations}\label{cf-apx}
		\begin{align}
			\Psi_{X_{k}\mid\tau}(t)
			&=\frac{\det\left(\mathbf{M}_{\mathrm{n}}(t)\right)^{-1}}{K^{3U_{1}}}\sum_{\boldsymbol{\upsilon}_{1}\in\mathcal{K}^{3}}\cdots\sum_{\boldsymbol{\upsilon}_{U_{1}}\in\mathcal{K}^{3}}\exp\Bigg[-\mathrm{tr}\log\bigg(\mathbf{I}_{N}+\frac{t^{2}}{4}\sum_{j=1}^{U_{1}}\widetilde{\mathbf{B}}_{k\mid\tau,\boldsymbol{\upsilon}_{j}}(t)\bigg)\Bigg]\label{cf-1}\\
			&\approx\frac{\det\left(\mathbf{M}_{\mathrm{n}}(t)\right)^{-1}}{K^{3U_{1}}}\sum_{\boldsymbol{\upsilon}_{1}\in\mathcal{K}^{3}}\cdots\sum_{\boldsymbol{\upsilon}_{U_{1}}\in\mathcal{K}^{3}}\exp\Bigg[-\frac{t^{2}}{4}\sum_{j=1}^{U_{1}}\mathrm{tr}\left(\widetilde{\mathbf{B}}_{k\mid\tau,\boldsymbol{\upsilon}_{j}}(t)\right)+\frac{1}{2}\left(\frac{t^{2}}{4}\right)^{2}\mathrm{tr}\bigg(\sum_{j=1}^{U_{1}}\widetilde{\mathbf{B}}_{k\mid\tau,\boldsymbol{\upsilon}_{j}}(t)\bigg)^{2}\Bigg]\label{cf-2}\\
			&=\frac{\det\left(\mathbf{M}_{\mathrm{n}}(t)\right)^{-1}}{K^{3U_{1}}}\sum_{\boldsymbol{\upsilon}_{1}\in\mathcal{K}^{3}}\cdots\sum_{\boldsymbol{\upsilon}_{U_{1}}\in\mathcal{K}^{3}}\exp\Bigg[-\frac{t^{2}}{4}\sum_{j=1}^{U_{1}}\mathrm{tr}\left(\widetilde{\mathbf{B}}_{k\mid\tau,\boldsymbol{\upsilon}_{j}}(t)\right)+\frac{1}{2}\left(\frac{t^{2}}{4}\right)^{2}\sum_{j=1}^{U_{1}}\mathrm{tr}\left(\widetilde{\mathbf{B}}_{k\mid\tau,\boldsymbol{\upsilon}_{j}}^{2}(t)\right)\notag\\
			&\qquad+\left(\frac{t^{2}}{4}\right)^{2}\sum_{1\leq j<r\leq U_{1}}\mathrm{tr}\left(\widetilde{\mathbf{B}}_{k\mid\tau,\boldsymbol{\upsilon}_{j}}(t)\widetilde{\mathbf{B}}_{k\mid\tau,\boldsymbol{\upsilon}_{r}}(t)\right)\Bigg]\label{cf-3}\\
			&\approx\frac{\det\left(\mathbf{M}_{\mathrm{n}}(t)\right)^{-1}}{K^{3U_{1}}}\sum_{\boldsymbol{\upsilon}_{1}\in\mathcal{K}^{3}}\cdots\sum_{\boldsymbol{\upsilon}_{U_{1}}\in\mathcal{K}^{3}}\prod_{j=1}^{U_{1}}\exp\Bigg[-\frac{t^{2}}{4}\mathrm{tr}\left(\widetilde{\mathbf{B}}_{k\mid\tau,\boldsymbol{\upsilon}_{j}}(t)\right)+\frac{1}{2}\left(\frac{t^{2}}{4}\right)^{2}\mathrm{tr}\left(\widetilde{\mathbf{B}}_{k\mid\tau,\boldsymbol{\upsilon}_{j}}^{2}(t)\right)\Bigg]\label{cf-4}\\
			&\approx\frac{\det\left(\mathbf{M}_{\mathrm{n}}(t)\right)^{-1}}{K^{3U_{1}}}\sum_{\boldsymbol{\upsilon}_{1}\in\mathcal{K}^{3}}\cdots\sum_{\boldsymbol{\upsilon}_{U_{1}}\in\mathcal{K}^{3}}\prod_{j=1}^{U_{1}}\det\left(\mathbf{I}_{N}+\frac{t^{2}}{4}\widetilde{\mathbf{B}}_{k\mid\tau,\boldsymbol{\upsilon}_{j}}(t)\right)^{-1}\label{cf-5}\\
			&=\det\left(\mathbf{M}_{\mathrm{n}}(t)\right)^{-1}\left[\frac{1}{K^{3}}\sum_{\boldsymbol{\upsilon}\in\mathcal{K}^{3}}\det\left(\mathbf{I}_{N}+\frac{t^{2}}{4}\widetilde{\mathbf{B}}_{k\mid\tau,\boldsymbol{\upsilon}}(t)\right)^{-1}\right]^{U_{1}}\label{cf-6}
		\end{align}
	\end{subequations}
	\hrule
\end{figure*}

\eqref{cf-2} follows from the second-order expansion of the matrix logarithm. \eqref{cf-3} expands the quadratic term, while \eqref{cf-4} neglects the second-order cross terms between different interfering EDs. \eqref{cf-5} follows by applying the same second-order approximation to each determinant. Since the discrete variables of different interfering EDs are independently and identically distributed, the multiple summations in \eqref{cf-5} can be factorized to obtain \eqref{cf-6}.  Finally, averaging \eqref{cf-6} over the uniformly distributed $\tau$ gives \eqref{cf-X}, which completes the proof.\qed

\section{Proof of {\it Lemma~\ref{l4}}}\label{app:d}
The asymptotically independent in the upper-tail of $k$ and $\tilde{k}$-th bins can be represented as \cite[Section 8.4]{app:d-1}
\begin{equation}
	\lim_{x\to\infty}\Pr(X_{\tilde{k}}>x\mid{X}_{k}>x)=0,	
\end{equation}
which is equal to
$\lim_{x\to\infty}\frac{\Pr\left(X_{\tilde{k}}>x,X_{k}>x\right)}{\bar{F}_{X_{k}}(x)}=0
$, where $\bar{F}_{X_{k}}(x)$ denotes the marginal complementary CDF (CCDF) of $k$-th undesired bin. Note that $\{X_{\tilde{k}}>x,X_k>x\}\subseteq\{X_{\tilde{k}}+X_k>2x\}$, therefore its sufficiently to prove
\begin{equation}
	\lim_{x\to\infty}\frac{\Pr\left(X_{k}+X_{\tilde{k}}>2x\right)}{\bar{F}_{X_k}(x)}=0.
\end{equation}
Then, we separately characterize the marginal distribution in the denominator and the joint distribution involved in the numerator.

We first consider the marginal distribution of a single bin. 

Given $\tau$ and $\{\boldsymbol{\upsilon}_{j}\}_{j=1}^{U_1}$ and based on the exact conditional CF in \eqref{cf-h}, define
$\mathbf{A}_{k\mid\tau,\{\boldsymbol{\upsilon}_{j}\}_{j=1}^{U_{1}}}=N_{0}\Omega\mathbf{R}/N+\sum_{j=1}^{U_{1}}\mathbf{B}_{k\mid\tau,\boldsymbol{\upsilon}_{j}}$.
Then, \eqref{cf-h} can be equivalently rewritten as
\begin{equation}\label{cf-xk}
	\begin{split}
		\Psi_{X_{k}\mid\tau,\{\boldsymbol{\upsilon}_{j}\}_{j=1}^{U_{1}}}(t)
		&=\det\bigg(\mathbf{M}_{\mathrm{n}}(t)+\frac{t^{2}}{4}\sum_{j=1}^{U_{1}}\mathbf{B}_{k\mid\tau,\boldsymbol{\upsilon}_{j}}\bigg)^{-1}\\
		&=\det\left(\mathbf{I}_{N}+\frac{t^{2}}{4}\mathbf{A}_{k\mid\tau,\{\boldsymbol{\upsilon}_{j}\}_{j=1}^{U_{1}}}\right)^{-1}.
	\end{split}
\end{equation}
Let $\lambda_{k,m\mid\tau,\{\boldsymbol{\upsilon}_{j}\}_{j=1}^{U_{1}}}$ denote the $m$-th eigenvalue of $\mathbf{A}_{k\mid\tau,\{\boldsymbol{\upsilon}_{j}\}_{j=1}^{U_{1}}}$. Applying the eigenvalue decomposition gives
\begin{equation}
	\begin{split}
		\Psi_{X_{k}\mid\tau,\{\boldsymbol{\upsilon}_{j}\}_{j=1}^{U_{1}}}(t)
		=\prod_{m=1}^{N}\bigg(1+\frac{\lambda_{k,m\mid\tau,\{\boldsymbol{\upsilon}_{j}\}_{j=1}^{U_{1}}}t^{2}}{4}\bigg)^{-1}.
	\end{split}
\end{equation}
Hence, $X_{k}$ has the same distribution as a sum of independent zero-mean Laplace RVs \cite[Chapter 26]{app:b-1}. Therefore, the upper tail is dominated by its largest eigenvalue. Defining $\lambda^{\star}_{k\mid\tau,\{\boldsymbol{\upsilon}_{j}\}_{j=1}^{U_{1}}}=\lambda_{\max}(\mathbf{A}_{k\mid\tau,\{\boldsymbol{\upsilon}_{j}\}_{j=1}^{U_{1}}})$. Since the upper-tail of a Laplace RV is exponential, we have
\begin{equation}\label{tail-x}
	\begin{split}
		\lim_{x\to\infty}\frac{1}{x}\log\Pr\left(X_{k}>x\mid\tau,\{\boldsymbol{\upsilon}_{j}\}^{U_{1}}_{j=1}\right)=\frac{-2}{\sqrt{\lambda^{\star}_{k\mid\tau,\{\boldsymbol{\upsilon}_{j}\}^{U_{1}}_{j=1}}}}.
	\end{split}
\end{equation}
Recall that $\mathbf{B}_{k\mid\tau,\boldsymbol{\upsilon}_{j}}
=
E_{\mathrm{s}}\Omega^{2}
\mathbf{R}^{\frac{1}{2}}
\mathbf{D}_{k\mid\tau,\boldsymbol{\upsilon}_{j}}
\mathbf{R}
\mathbf{D}_{k\mid\tau,\boldsymbol{\upsilon}_{j}}^{\mathrm{H}}
\mathbf{R}^{\frac{1}{2}}$, the diagonal entries of $\mathbf{D}_{k\mid\tau,\boldsymbol{\upsilon}_j}$ satisfy
\begin{equation}
	\begin{split}
		\big|[\mathbf{D}_{k\mid\tau,\boldsymbol{\upsilon}_j}]_{l,l}\big|&\leq\sum_{n\in\mathcal{S}_{l\mid\tau}}\left|x_{\boldsymbol{\upsilon}}[n]\right|\left|x_{k}[n]\right|=\frac{Q}{K}=\frac{1}{N},
	\end{split}
\end{equation}
which gives $\left\|\mathbf{D}_{k\mid\tau,\boldsymbol{\upsilon}_j}\right\|_{2}\leq1/N$ and we have $\left\|\mathbf{B}_{k\mid\tau,\boldsymbol{\upsilon}_j}\right\|_{2}\leq E_{\mathrm{s}}\Omega^{2}
\|\mathbf{R}^{\frac{1}{2}}\|_{2}^{2}
\left\|\mathbf{R}\right\|_{2}
\left\|\mathbf{D}_{k\mid\tau,\boldsymbol{\upsilon}_j}\right\|_{2}^{2}\leq\frac{E_{\mathrm{s}}\Omega^{2}}{N^{2}}
\left(\rho^{\star}\right)^{2}$, where $\rho^{\star}=\lambda_{\max}(\mathbf{R})$. We then have
\begin{equation}\label{lam-1}
	\begin{split}
		\lambda^\star_{k\mid\tau,\{\boldsymbol{\upsilon}_{j}\}^{U_{1}}_{j=1}}\leq\frac{N_{0}\Omega}{N}\rho^{\star}+\frac{E_{\mathrm{s}}\Omega^{2}U_{1}}{N^{2}}(\rho^{\star})^{2}=\Lambda_1.
	\end{split}
\end{equation}
The equality is attained when $d_{j}=0$ and $a_{\mathrm{p},j}=a_{\mathrm{q},j}=k$ for all interfering EDs, under which $\mathbf{D}_{k\mid\tau,\boldsymbol{\upsilon}_{j}}=\mathbf{I}_{N}/N$. Since the discrete state space is finite and the above state has a nonzero probability, then \eqref{tail-x} becomes
\begin{equation}\label{tail-u}
	\begin{split}
		\lim_{x\to\infty}\frac{1}{x}\log\bar{F}_{X_k}(x)=-\frac{2}{\sqrt{\Lambda_{1}}}.
	\end{split}
\end{equation}

We next consider two distinct undesired bins $k$ and $\tilde{k}$.

Similar to the above process, the CF of  $X_{\tilde{k}}+X_k$ have the same form as \eqref{cf-xk}, by substituting $\mathbf{A}_{k\mid\tau,\{\boldsymbol{\upsilon}_{j}\}_{j=1}^{U_{1}}}$ with 
\begin{equation}
	\begin{split}
		&\tilde{\mathbf{A}}_{\tilde{k},k\mid\tau,\{\boldsymbol{\upsilon}_{j}\}_{j=1}^{U_{1}}}=N_{0}\Omega\mathbf{R}^{\frac{1}{2}}\mathbf{T}_{\tilde{k},k\mid\tau}\mathbf{R}^{\frac{1}{2}}\\
		&\quad+E_{\mathrm{s}}\Omega^{2}\sum_{j=1}^{U_{1}}\mathbf{R}^{\frac{1}{2}}\mathbf{D}_{\tilde{k},k\mid\tau,\boldsymbol{\upsilon}_{j}}\mathbf{R}\mathbf{D}_{\tilde{k},k\mid\tau,\boldsymbol{\upsilon}_{j}}^{\mathrm{H}}\mathbf{R}^{\frac{1}{2}}.
	\end{split}
\end{equation}
where $\mathbf{D}_{\tilde{k},k\mid\tau,\boldsymbol{\upsilon}_j}
=
\mathbf{D}_{\tilde{k}\mid\tau,\boldsymbol{\upsilon}_j}
+
\mathbf{D}_{k\mid\tau,\boldsymbol{\upsilon}_j}$ and $\mathbf{T}_{\tilde{k},k\mid\tau}=\mathrm{diag}\big\{\sum_{n\in\mathcal{S}_{l\mid\tau}}\left|x_{\tilde{k}}[n]+x_{k}[n]\right|^{2}\big\}_{l=1}^{N}$. For $Q>1$, we have
\begin{equation}
	\begin{split}
		&\big\|\mathbf{D}_{\tilde{k},k\mid\tau,\boldsymbol{\upsilon}_j}\big\|_{2}
		=\max_{l}
		\Big|
		\sum_{n\in\mathcal{S}_{l\mid\tau}}
		x_{\boldsymbol{\upsilon}}[n]
		\left(x_{\tilde{k}}[n]+x_{k}[n]\right)^{*}
		\Big|\\
		&\quad\leq
		\max_{l}
		\sum_{n\in\mathcal{S}_{l\mid\tau}}
		\left|x_{\boldsymbol{\upsilon}}[n]\right|
		\left|x_{\tilde{k}}[n]+x_{k}[n]\right|
		<\frac{2Q}{K}
		=\frac{2}{N}.
	\end{split}
\end{equation}
and 
\begin{equation}
	\begin{split}
		\big\|\mathbf{T}_{\tilde{k},k\mid\tau}\big\|_{2}=\max_{l}
		\sum_{n\in\mathcal{S}_{l\mid\tau}}
		\left|x_{\tilde{k}}[n]+x_{k}[n]\right|^{2}<\frac{4Q}{K}
		=\frac{4}{N}.
	\end{split}
\end{equation}
 Defining $\tilde{\lambda}^{\star}_{\tilde{k},k\mid\tau,\{\boldsymbol{\upsilon}_{j}\}_{j=1}^{U_{1}}}$ $=\lambda_{\max}(\tilde{\mathbf{A}}_{\tilde{k},k\mid\tau,\{\boldsymbol{\upsilon}_{j}\}_{j=1}^{U_{1}}})$, its conditional upper-tail exponent is
\begin{equation}\label{tail-xx}
	\begin{split}
		\!\lim_{x\to\infty}\!\frac{1}{x}\!\log\Pr\!\big(X_{\tilde{k}}+X_{k}\!\!>\!2x\!\mid\!\tau,\{\boldsymbol{\upsilon}_{j}\}_{j=1}^{U_{1}}\big)\!=\!\frac{-4}{\sqrt{\tilde{\lambda}_{\tilde{k},k\mid\tau,\{\boldsymbol{\upsilon}_{j}\}_{j=1}^{U_{1}}}^{\star}\!\!}}.
	\end{split}
\end{equation}
Since all involved indices take values in finite sets, there exist constants $\mu<2$ and $\nu<4$ such that
$\|\mathbf{D}_{\tilde{k},k\mid\tau,\boldsymbol{\upsilon}_j}\|_{2}\leq\frac{\mu}{N}$ and $\|\mathbf{T}_{\tilde{k},k\mid\tau}\|_{2}\leq\frac{\nu}{N}$. Accordingly, we have
\begin{equation}\label{lambda12}
	\begin{split}
		\Lambda_{2}&\triangleq\max_{\substack{\tilde{k}<k,\tilde{k},k\in\mathcal{K}\setminus\{a_{0}\}\\\tau,\{\boldsymbol{\upsilon}_{j}\}_{j=1}^{U_{1}}}}\lambda_{\max}\Big(\tilde{\mathbf{A}}_{\tilde{k},k\mid\tau,\{\boldsymbol{\upsilon}_{j}\}_{j=1}^{U_{1}}}\Big)\\
		&\leq\frac{\nu N_{0}\Omega}{N}\rho^{\star}+\frac{\mu^{2}E_{\mathrm{s}}\Omega^{2}U_{1}}{N^{2}}\left(\rho^{\star}\right)^{2}<4\Lambda_{1}.
	\end{split}
\end{equation}
Combining \eqref{tail-u}, \eqref{tail-xx} and \eqref{lambda12}, we have
\begin{equation}
	\begin{split}
		\limsup_{x\to\infty}\frac{1}{x}\log\frac{\Pr\left(X_{\tilde{k}}+X_{k}>2x\right)}{\bar{F}_{X_k}(x)}\leq-\frac{4}{\sqrt{\Lambda_{2}}}+\frac{2}{\sqrt{\Lambda_{1}}}<0,
	\end{split}
\end{equation}
Note that $\limsup_{x\to\infty}\log(f(x))/x<0$ is sufficient for $\lim_{x\to\infty}f(x)=0$. This completes the proof.\qed

\section{Numerical Evaluation of Kendall's Tau}\label{app:e}

Similar to \eqref{cdf-exc}, directly calculating \eqref{kendall-tau} requires averaging is numerically intractable. Therefore, we first transform all variables into a uniform vector $\mathbf{U}\in[0,1]^{D}$ and then using the deterministic numerical integration, where
\begin{equation}
	D=\underbrace{2NU}_{\mathrm{complex}\:\mathrm{channels}}+\underbrace{2K}_{\mathrm{complex}\:\mathrm{noise}}+\underbrace{1}_{\tau}+\underbrace{3U_{1}}_{\{\boldsymbol{\upsilon}_{j}\}_{j=1}^{U_{1}}}.
\end{equation}
Based on the definition, Kendall's tau can then be written as \eqref{kendall},
\begin{figure*}[t]
	\begin{subequations}\label{kendall}
		\begin{align}
			\tau_{\mathrm{K}}
			={}&
			\mathrm{E}_{\mathbf{U},\widetilde{\mathbf{U}}}\left[
			\operatorname{sgn}\left(
			\left(\varpi_{0}[a_{0};\mathbf{U}]
			-\varpi_{0}[a_{0};\widetilde{\mathbf{U}}]\right)
			\left(X^{\star}(\mathbf{U})
			-X^{\star}(\widetilde{\mathbf{U}})\right)
			\right)
			\right]
			\label{kendall-a}\\
			={}&
			\int_{[0,1]^{D}}\int_{[0,1]^{D}}
			\operatorname{sgn}\left[
			\left(\varpi_{0}[a_{0};\mathbf{u}]
			-\varpi_{0}[a_{0};\widetilde{\mathbf{u}}]\right)
			\left(X^{\star}(\mathbf{u})
			-X^{\star}(\widetilde{\mathbf{u}})\right)
			\right]
			\mathrm{d}\mathbf{u}\mathrm{d}\widetilde{\mathbf{u}}
			\label{kendall-b}\\
			\approx{}&
			\frac{1}{2L_0(2L_0-1)}
			\sum_{\substack{\iota,\widetilde{\iota}=1\\
					\iota\neq\widetilde{\iota}}}^{L}
			\operatorname{sgn}\left[
			\left(\varpi_{0}^{(\iota)}[a_{0}]
			-\varpi_{0}^{(\widetilde{\iota})}[a_{0}]\right)
			\left(X_{\iota}^{\star}
			-X_{\widetilde{\iota}}^{\star}\right)
			\right]
			\label{kendall-c}\\
			={}&
			\frac{2}{2L_0(2L_0-1)}
			\sum_{1\leq\iota<\widetilde{\iota}\leq L}
			\operatorname{sgn}\left[
			\left(\varpi_{0}^{(\iota)}[a_{0}]
			-\varpi_{0}^{(\widetilde{\iota})}[a_{0}]\right)
			\left(X_{\iota}^{\star}
			-X_{\widetilde{\iota}}^{\star}\right)
			\right].
			\label{kendall-d}
		\end{align}
	\end{subequations}
	\hrule
\end{figure*}
where \eqref{kendall-a} follows from the uniform transformation of all underlying variables. \eqref{kendall-b} follows since $\mathbf{U}$ and $\widetilde{\mathbf{U}}$ are independently and uniformly distributed over $[0,1]^{D}$. \eqref{kendall-c} is obtained by deterministic numerical integration using $2L_0$ points generated from a Kronecker sequence, while \eqref{kendall-d} follows from the symmetry between $(\iota,\widetilde{\iota})$ and $(\widetilde{\iota},\iota)$. \eqref{kendall-d} can be numerically evaluated using Algorithm~\ref{alg:delta}. To avoid the exponential complexity of a tensor-product uniform grid, the integration points are generated from a Kronecker sequence, which provides a deterministic and approximately uniform coverage of the unit hypercube\cite{app:e-1}.
\balance


\begin{thebibliography}{1}
	\bibliographystyle{IEEEtran}
	\bibitem{bibitem1-1}
	M. ~Jouhari, N.~Saeed, M.-S.~Alouini and E.~M.~Amhoud, ``A Survey on Scalable LoRaWAN for Massive IoT: Recent Advances, Potentials, and Challenges,'' \emph{IEEE Commun. Surveys Tuts.}, vol. 25, no. 3, pp. 1841-1876, thirdquarter 2023.
	\bibitem{bibitem1-2}
	J.~P.~Shanmuga Sundaram, W.~Du and Z.~Zhao, ``A Survey on LoRa Networking: Research Problems, Current Solutions, and Open Issues,'' \emph{IEEE Commun. Surveys Tuts.}, vol. 22, no. 1, pp. 371-388, Firstquarter 2020.
	\bibitem{bibitem1-3}
	Z.~Liang, G.~Cai, J.~He, G.~Kaddoum, C.~Huang and M.~Debbah, ``RIS-Enabled Anti-Interference in LoRa Systems,'' \emph{IEEE Trans. Commun.}, vol. 72, no. 10, pp. 6599-6616, Oct. 2024.
	\bibitem{bibitem1-4}
	L.~Beltramelli, A.~Mahmood, P.~\"{O}sterberg and M.~Gidlund, ``LoRa Beyond ALOHA: An Investigation of Alternative Random Access Protocols,'' \emph{IEEE Trans. Ind. Inf.}, vol. 17, no. 5, pp. 3544-3554, May 2021.
	\bibitem{bibitem1-5}
	G.~Mu {\em et al.}, ``Coverage Performance Analysis of FAS-enhanced LoRa Wide Area Networks under both Co-SF and Inter-SF Interference,'' {\em IEEE Trans. Veh. Technol.}, early access, \url{DOI: 10.1109/TVT.2026.3696818}.
	\bibitem{bibitem1-6}
	S.~A.~Tegos, P.~D.~Diamantoulakis, A.~S.~Lioumpas, P.~G.~Sarigiannidis and G.~K.~Karagiannidis, ``Slotted ALOHA With NOMA for the Next Generation IoT,'' {\em IEEE Trans. Commun.}, vol. 68, no. 10, pp. 6289-6301, Oct. 2020.
	\bibitem{bibitem1-7}
	O.~Afisiadis, M.~Cotting, A.~Burg and A.~Balatsoukas-Stimming, ``On the Error Rate of the LoRa Modulation With Interference,'' {\em IEEE Trans. Wireless Commun.}, vol. 19, no. 2, pp. 1292-1304, Feb. 2020
	\bibitem{bibitem1-8}
	M.~Xhonneux, J.~Tapparel, A.~Balatsoukas-Stimming, A.~Burg and O.~Afisiadis, ``A Maximum-Likelihood-Based Two-User Receiver for LoRa Chirp Spread-Spectrum Modulation,'' {\em IEEE Internet Things J.}, vol. 9, no. 22, pp. 22993-23007, 15 Nov.15, 2022.
	\bibitem{bibitem1-9}
	M.~Heusse, C.~Caillouet and A.~Duda, ``Performance of Unslotted ALOHA With Capture and Multiple Collisions in LoRaWAN,'' {\em IEEE Internet Things J.}, vol. 10, no. 20, pp. 17824-17838, 15 Oct.15, 2023.
	\bibitem{bibitem1-10}
	O. Georgiou and U. Raza, ``Low Power Wide Area Network Analysis: Can LoRa Scale?,'' {\em IEEE Wireless Commun. Lett.}, vol. 6, no. 2, pp. 162-165, April 2017.
	\bibitem{bibitem1-11}
	Z.~Cao, Y.~Hu, H.~Jin and J.-B.~Seo, ``NOMA-Aided Pure ALOHA With Immediate Collision Resolution for Low-Power IoT Communications,'' {\em IEEE Internet Things J.,} vol. 12, no. 24, pp. 52658-52674, 15 Dec.15, 2025.
	\bibitem{bibitem1-12}
	Q. Yu, H. Wang, D. He and Z. Lu, ``Enhanced Group-Based Chirp Spread Spectrum Modulation: Design and Performance Analysis,'' {\em  IEEE Internet Things J.}, vol. 12, no. 5, pp. 5079-5092, 1 March1, 2025.
	\bibitem{bibitem1-13}
	A. W. Azim, R. Shubair and M. Chafii, ``Chirp Spread Spectrum-Based Waveform Design and Detection Mechanisms for LPWAN-Based IoT: A Survey,'' {\em IEEE Access}, vol. 12, pp. 24949-25017, 2024.
	\bibitem{bibitem1-14}
	I.~Bizon~Franco~de~Almeida, M.~Chafii, A.~Nimr and G.~Fettweis, ``Alternative Chirp Spread Spectrum Techniques for LPWANs,'' {\em IEEE Trans. Green Commun. Networking}, vol. 5, no. 4, pp. 1846-1855, Dec. 2021.
	\bibitem{bibitem1-15}
	R.~Bomfin, M.~Chafii, and G.~Fettweis. ``A novel modulation for IoT:
	PSK-LoRa'',  {\em In Proc. IEEE Veh. Technol. Conf.}, Apr. 2019, pp. 1-5.
	\bibitem{bibitem1-16}
	T.~Elshabrawy and J.~Robert, ``Interleaved Chirp Spreading LoRa-Based Modulation,'' {\em IEEE Internet Things J.}, vol. 6, no. 2, pp. 3855--3863, Apr. 2019.
	\bibitem{bibitem1-17}
	L.~Wang and M.~M.~Wang, ``Interleaved Chirp Spreading LoRa-Based Modulation With Continuous Phase,'' {\em IEEE Communi. Lett.}, vol. 29, no. 5, pp. 873-877, May 2025.
	\bibitem{bibitem1-18}
	A.~W.~Azim, A.~Bazzi, R.~Bomfin, R.~Shubair and M.~Chafii, ``Layered Chirp Spread Spectrum Modulations for LPWANs,'' {\em IEEE Trans. Commun.}, vol. 72, no. 3, pp. 1671-1687, March 2024.
	\bibitem{bibitem1-19}
	J.~Zhou, J.~Niu, B.~Li, T.~Wei and G.~Mu, ``RIS-Aided LoRa Uplink Transmission With Differential Index Modulation,'' {\em IEEE Trans. Commun.}, vol. 74, pp. 5893-5906, 2026.
	\bibitem{bibitem1-20}
	X.~Yu, X.~Cai, W.~Xu, H.~Sun and L.~Wang, ``Differential Phase Shift Keying-Aided Multimode Chirp Spread Spectrum Modulation,'' {\em IEEE Wireless Commun. Lett.}, vol. 13, no. 2, pp. 298--302, Feb. 2024.
	\bibitem{bibitem1-21}
	G.~Baruffa and L.~Rugini, ``Performance of LoRa-based schemes and
	quadrature chirp index modulation,''{\em IEEE Internet Things J.}, vol. 9, no. 10, pp. 7759–7772, May 2022.
	\bibitem{bibitem1-22}
	M.~A.~Ben~Temim, G.~Ferré, B.~Laporte-Fauret, D.~Dallet, B.~Minger and L.~Fuché, ``An Enhanced Receiver to Decode Superposed LoRa-Like Signals,'' {\em IEEE Internet Things J.}, vol. 7, no. 8, pp. 7419-7431, Aug. 2020.
	\bibitem{bibitem1-23}
	J.-B.~Seo, Y.~Hu, H.~Jin and S.~De, ``ALOHA With SIC-Aided Collision Resolution,'' {\em in IEEE Internet Things J.}, vol. 12, no. 8, pp. 10194-10209, 15 April15, 2025.
	\bibitem{bibitem1-24}
	D.-T.~Ta, K.~Khawam, S.~Lahoud, C.~Adjih and S.~Martin, ``LoRa-MAB: Toward an Intelligent Resource Allocation Approach for LoRaWAN,'' in {\em Proc. IEEE Global Commun. Conf. (GLOBECOM)}, Waikoloa, HI, USA, Dec. 2019, pp. 1-6.
	\bibitem{bibitem1-25}
	K.-K.~Wong {\em et al.}, ``Fluid antenna systems,'' \emph{IEEE Trans. Wireless Commun.}, vol.~20, no.~3, pp. 1950--1962, Mar. 2021.
	\bibitem{bibitem1-26}
	T. Wu {\em et al.}, ``Fluid antenna systems enabling 6G: Principles, applications, and research directions,'' \emph{IEEE Wireless Commun.}, , early access, \url{ 10.1109/MWC.2025.3629597.}
	\bibitem{bibitem1-27}
	X.~Pi, L.~Zhu, H.~Mao, Z.~Xiao, X.-G.~Xia and R.~Zhang, ``6D Movable Antenna Enhanced Multi-Access Point Coordination via Position and Orientation Optimization,'' \emph{IEEE Trans. Wireless Commun.}, vol. 25, pp. 915-930, 2026.
	\bibitem{bibitem1-28}
	Z. Xiao {\it et al.}, ``Movable Antenna Aided NOMA: Joint Antenna Positioning, Precoding, and Decoding Design,'' \emph{IEEE Trans. Wireless Commun.}, vol. 25, pp. 4595-4612, 2026
	\bibitem{bibitem1-29}
	K.-K.~Wong and K.-F.~Tong, ``Fluid antenna multiple access,'' \emph{IEEE Trans. Wireless Commun.}, vol.~21, no.~7, pp. 4801--4815, Jul. 2022.
	\bibitem{bibitem1-30}
	N.~Waqar, K.-K.~Wong, C.-B.~Chae, R.~Murch, ``Fast Fluid Antenna Multiple Access,'' \emph{arXiv preprint}, \url{arXiv:2605.23642}, May. 2026.
	\bibitem{bibitem1-31}
	K.-K.~Wong, D.~Morales-Jimenez, K.-F.~Tong and C.-B.~Chae, ``Slow Fluid Antenna Multiple Access,'' \emph{IEEE Transactions on Communications}, vol. 71, no. 5, pp. 2831-2846, May 2023.
	\bibitem{bibitem1-32}
	K.-K.~Wong, K.-F.~Tong, Y.~Chen, Y.~Zhang and C.-B.~Chae, ``Opportunistic Fluid Antenna Multiple Access,'' {\em IEEE Trans. Wireless Commun.}, vol. 22, no. 11, pp. 7819-7833, Nov. 2023.
	\bibitem{bibitem1-33}
	K.-K.~Wong, C.-B.~Chae and K.-F.~Tong, ``Compact Ultra Massive Antenna Array: A Simple Open-Loop Massive Connectivity Scheme,'' {\em IEEE Trans. Wireless Commun.}, vol. 23, no. 6, pp. 6279-6294, June 2024.
	\bibitem{bibitem1-34}
	N.~Waqar, K.-K.~Wong, C.-B.~Chae and R.~Murch, ``Turbocharging Fluid Antenna Multiple Access,'' {\em IEEE Trans. Wireless Commun.}, vol. 25, pp. 4038-4052, 2026.
	\bibitem{bibitem1-35}
	T.~Wu, J.~Zheng, X.~Lai, M.~Elkashlan, H.~Shin and N.~Al-Dhahir, ``Fluid Antennas Meet Intelligent Surfaces: Security Analysis of NOMA Systems Under Hardware Impairments,'' {\em IEEE Trans. Cognit. Commun. Netw.}, vol. 12, pp. 7775-7788, 2026.
	\bibitem{bibitem1-36}
	T.~Wu {\it et al.}, ``Unleashing More Potential From FAS: A Framework of FAS-CoNOMA Systems,'' {\em IEEE Trans. Commun.}, vol. 74, pp. 4820-4836, 2026.
	\bibitem{bibitem1-37}
	H.~Hong {\it et al.}, ``Downlink OFDM-FAMA in 5G-NR Systems,'' {\em IEEE Trans. Wireless Commun.}, vol. 24, no. 12, pp. 10116-10132, Dec. 2025.
	
	\bibitem{bibitem2-1}
	A.~Maleki, H.~H.~Nguyen, E.~Bedeer and R.~Barton ``A Tutorial on Chirp Spread Spectrum Modulation for LoRaWAN: Basics and Key Advances,'' {\em IEEE Open J. Commun. Soc.}, vol. 5, pp. 4578-4612, 2024.
	\bibitem{bibitem2-2}
	R.~Marini, K.~Mikhaylov, G.~Pasolini and C.~Buratti, ``Low-Power Wide-Area Networks: Comparison of LoRaWAN and NB-IoT Performance,'' {\em IEEE Internet Things J.}, vol. 9, no. 21, pp. 21051-21063, 1 Nov.1, 2022.
	\bibitem{bibitem2-3}
	L.~Amichi, M.~Kaneko, E.~H.~Fukuda, N.~El~Rachkidy and A.~Guitton, ``Joint Allocation Strategies of Power and Spreading Factors With Imperfect Orthogonality in LoRa Networks,'' {\em IEEE Trans. Commun.}, vol. 68, no. 6, pp. 3750-3765, June 2020.
	\bibitem{bibitem2-4}
	S.~Yu, X.~Xia, Z.~Zhang, N.~Hou and Y.~Zheng, ``FDLoRa: Scaling Downlink Concurrent Transmissions With Full-Duplex LoRa Gateways,''{\em IEEE Trans. Mob. Comput.}, vol. 24, no. 10, pp. 10668-10682, Oct. 2025.
	\bibitem{bibitem2-5}
	{\it TS001-1.0.4 LoRaWAN L2 1.0.4 Specification}, LoRa Alliance, Fremont, CA, USA, Oct. 2020.
	\bibitem{bibitem2-6}
	P.~Gkotsiopoulos, D.~Zorbas and C.~Douligeris, ``Performance Determinants in LoRa Networks: A Literature Review,'' {\em IEEE Commun. Surveys Tuts.}, vol. 23, no. 3, pp. 1721-1758, thirdquarter 2021.
	\bibitem{bibitem2-7}
	M.~Chiani, and A.~Elzanaty ``On the LoRa Modulation for IoT: Waveform Properties and Spectral Analysis,'' {\em IEEE Internet Things J.}, vol. 6, no. 5, pp. 8463–8470, Oct. 2019.
	\bibitem{bibitem2-8}
	W.~K.~New {\it et al.}, ``A Tutorial on Fluid Antenna System for 6G Networks: Encompassing Communication Theory, Optimization Methods and Hardware Designs,'' {\em IEEE Commun. Surveys Tuts.}, vol. 27, no. 4, pp. 2325-2377, Aug. 2025.
	\bibitem{bibitem2-9}
	J. Costantine, Y. Tawk, S. E. Barbin and C. G. Christodoulou, "Reconfigurable Antennas: Design and Applications," in Proceedings of the IEEE, vol. 103, no. 3, pp. 424-437, March 2015
	\bibitem{bibitem2-10}
	T.~Elshabrawy and J.~Robert, ``Closed-form approximation of LoRa modulation BER performance,'' {\em IEEE Commun. Lett.}, vol.~22, no.~9, pp. 1778--1781, Sep. 2018.
	
	\bibitem{bibitem3-1}
	F.~Rostami~Ghadi, K.-K.~Wong, F.~Javier~López-Martínez, C.-B.~Chae, K.-F.~Tong and Y.~Zhang, ``A Gaussian Copula Approach to the Performance Analysis of Fluid Antenna Systems,'' {\em IEEE Trans. Wireless Commun.}, vol. 23, no. 11, pp. 17573-17585, Nov. 2024.
	\bibitem{bibitem3-2}
	Y.~Hou {\em et al.}, ``A Copula-Based Approach to Performance Analysis of Fluid Antenna System With Multiple Fixed Transmit Antennas,'' {\em IEEE Wireless Commun. Lett.}, vol. 13, no. 2, pp. 501-504, Feb. 2024.
	
	\bibitem{bibitem4-1}
	L.~Amichi, M.~Kaneko, E.~H.~Fukuda, N.~El~Rachkidy and A.~Guitton, ``Joint Allocation Strategies of Power and Spreading Factors With Imperfect Orthogonality in LoRa Networks,'' {\em IEEE Trans. Commun.}, vol. 68, no. 6, pp. 3750-3765, June 2020.
	\bibitem{bibitem4-2}
	O. Georgiou and U. Raza, ``Low Power Wide Area Network Analysis: Can LoRa Scale?,'' {\em IEEE Wireless Commun. Lett.}, vol. 6, no. 2, pp. 162-165, April 2017.
	\bibitem{bibitem4-3}
	G.~Mu {\em et al.}, ``On Performance of LoRa Fluid Antenna Systems,'' {\em IEEE Trans. Wireless Commun.}, vol. 25, pp. 4869-4886, 2026.
	
	\bibitem{app:b-1}
	C.~Forbes, M.~Evans, N.~Hastings, and B.~Peacock, \emph{Statistical Distributions}, 4th ed. Hoboken, NJ, USA: Wiley, 2011.
	
	\bibitem{app:c-1}
	B.~Picinbono, `Second-order complex random vectors and normal distributions,'' {\em IEEE Trans. Signal Process.}, vol. 44, no. 10, pp. 2637-2640, Oct. 1996.
	
	\bibitem{app:d-1}
	S.~Coles, J.~Bawa, L.~Trenner, \emph{et al.}, ``An Introduction to Statistical Modeling of Extreme Values,'' Springer, London, U.K., 2001.
	
	\bibitem{app:e-1}
	H.~Niederreiter, \emph{Random Number Generation and Quasi-Monte Carlo Methods}. Philadelphia, PA, USA: Society for Industrial and Applied Mathematics, 1992.
\end{thebibliography}
\end{document}